\documentclass[submission,copyright,creativecommons,sharealike,noncommercial]{eptcs}
\providecommand{\event}{GCM 2019} 
\usepackage{breakurl}             
\usepackage{underscore}           

\usepackage[english]{babel}

\usepackage{graphicx}        
\usepackage{subcaption}

\usepackage{amsfonts}
\usepackage{multirow}
\usepackage{tikz}
\usetikzlibrary{patterns,%
								positioning,%
								calc,%
								scopes,%
								decorations.fractals,%
								shapes.geometric,%
								shadows,%
								decorations.pathreplacing,%
								decorations.pathmorphing,%
								decorations.markings,%
								decorations,%
								shapes,%
								arrows,
                                                                automata,petri,
                                                                lindenmayersystems,
}
\usepackage{xifthen}

\usepackage{latexsym,amsthm,amssymb,amsmath}

\newtheorem{definition}{Definition}
\newtheorem{remark}{Remark}
\newtheorem{example}{Example}
\newtheorem{proposition}{Proposition}
\newtheorem{theorem}{Theorem}
\newtheorem{lemma}{Lemma}

\title{Transformation of Turing Machines into\\ Context-Dependent Fusion Grammars}

\author{Aaron Lye
\institute{University of Bremen, Department of Computer Science and Mathematics\\ 
P.O.Box 33 04 40, 28334 Bremen, Germany}
\email{lye@math.uni-bremen.de}
}
\def\titlerunning{Transformation of Turing Machines into Context-Dependent Fusion Grammars}
\def\authorrunning{A. Lye}

\newcommand{\dder}{\mathop{\Longrightarrow}\limits}

\newcommand{\N}{\mathbb N}

\newcommand{\TM}{\mathit{TM}}

\newcommand{\fuscomp}{\overline}

\newcommand{\fr}{\mathit{fr}}
\newcommand{\cdfr}{\mathit{cdfr}}

\newcommand{\CDFG}{\mathit{CDFG}}

\newcommand{\att}{\mathit{att}}
\newcommand{\lE}{\mathit{lab}}
\newcommand{\sE}{s}
\newcommand{\tE}{t}

\newcommand{\HSigma}{\mathcal{H}_{\Sigma}}
\newcommand{\HT}{\mathcal{H}_{T}}
\newcommand{\Ht}[1]{\mathcal{H}_{#1}}
\newcommand{\HTM}{\mathcal{H}_{T \cup M}}
\newcommand{\Htm}[2]{\mathcal{H}_{#1 + #2}}
\newcommand{\calC}{\mathcal{C}}

\newcommand{\tmtapestart}{\vartriangleright}
\newcommand{\tmtapeend}{\vartriangleleft}
\newcommand{\tmblank}{\square}
\newcommand{\tmstep}[1]{\vdash_{#1}}

\newcommand{\sg}{\mathit{sg}}
\newcommand{\tg}{\mathit{tg}}
\newcommand{\encsg}[1]{_{\tmtapestart}\mathit{sg}#1_{\tmtapeend}}
\newcommand{\fsym}[1]{#1_f}

\begin{document}
\maketitle

\begin{abstract}
Context-dependent fusion grammars were recently introduced as devices for the generation of hypergraph languages.
In this paper, we show that this new type of hypergraph grammars, where the application of fusion rules is restricted by positive and negative context conditions, is a universal computation model.
Our main result is that Turing machines can be transformed into these grammars such that the recognized language of the Turing machine and the generated language of the corresponding context-dependent fusion grammar coincide up to representation of strings as graphs. 
As a corollary we get that context-dependent fusion grammars can generate all recursively enumerable string languages.

\end{abstract}

\section{Introduction}
\label{sec:introduction}
In 2017 we introduced fusion grammars as generative devices on hypergraphs~\cite{Kreowski-Kuske-Lye:17a}.
They are motivated by the observation, that one encounters various fusion processes in various scientific fields like DNA computing, chemistry, tiling, fractal geometry, visual modeling and others.
The common principle is, that a few small entities may be copied and fused to produce more complicated entities.
However, it seems that the generative power of fusion grammars (without context-conditions or regulations) is limited (cf.~\cite{Kreowski-Kuske-Lye:17a,Lye:18}). 
Furthermore, there are numerous examples of fusion processes restricted to certain conditions, e.g. the presence of enzymes accelerating chemical reactions.
In~\cite{Kreowski-Kuske-Lye:19a} we introduced context-dependent fusion grammars as a generalization of fusion grammars to simulate Petri nets.
It turns out, that context-dependent fusion grammars are powerful enough to simulate Turing machines.
We construct a transformation of Turing machines into
context-dependent fusion grammars in such a way that the recognized language of the Turing machine and the language generated by the corresponding grammar coincide up to representation of strings as graphs.\footnote{Instead of Turing machines some other equivalent computational formalism could be chosen, e.g. Petri nets with inhibitor arcs which would be an extension of the transformation presented in~\cite{Kreowski-Kuske-Lye:19a}. In the considered approaches the transformations were technical and of similar complexity.
Turing machines have the advantage of being an established and well known computational model.}
As a corollary we get that context-dependent fusion grammars can generate all recursively enumerable string languages (up to representation) and that they are universal in this respect.

Relating computational models to Turing machines is an old and established approach which can be found in most foundations textbooks in theoretical computer science.
Moreover, it is well known that graph transformation in general is Turing-complete. In 1978 Uesu presented a system of graph grammars that generates all recursively enumerable sets of labeled graphs (cf.~\cite{Uesu:78}).
In~\cite{Ehrig-Habel-Kreowski:92} Ehrig et.~al. presented a transformation of Chomsky grammars in graph grammars (cf. also \cite{Kreowski:93b,Kreowski-Klempien-Hinrichs-Kuske:2005} for similar results).
Furthermore, asking ``what programming constructs are needed on top of graph transformation rules to obtain a computationally complete language''~\cite{Habel-Plump:01} is not a new question. In~\cite{Habel-Plump:01} Habel and Plump presented a graph program that simulates a Turing machine.
Due to the novelty of (variants of) fusion grammars, there are many open questions.
Enhancing fusion grammars by the inversion of fusions led to the introduction of splitting/fusion grammars in~\cite{Kreowski-Kuske-Lye:18b}.
It is shown that splitting/fusion grammars  can simulate Chomsky grammars and connective hypergraph grammars.

Our construction differs significantly from those cited above due to the semantics of context-dependent fusion grammars.
A context-dependent fusion grammar provides a start hypergraph and a finite set of fusion labels (besides some markers and terminals).
The fusion labels have complements and serve as rules.
A context-dependent fusion is defined by choosing two complementarily labeled hyperedges provided that certain positive and negative context conditions are satisfied, removing them and merging the corresponding attachment vertices.
Given a hypergraph, the set of all possible fusions is finite as fusions never create anything.
To overcome this limitation, we allow arbitrary multiplications of connected components, i.e., connected subhypergraphs of maximal size, within derivations in addition to fusion.
All modifications must be expressed in this way.

The paper is organized as follows.
In Section~\ref{sec:preliminaries}, basic notions and
notations of hypergraphs are recalled.
Section~\ref{sec:turing-machines} and~\ref{sec:fusion-with-context} recall the notions of Turing machines and context-dependent fusion grammars, respectively.
Section~\ref{sec:trans-tm-to-cdfg} presents the reduction of Turing machines to
context-dependent fusion grammars.
Section~\ref{sec:conclusion} concludes the paper pointing out some open problems.

\section{Preliminaries}
\label{sec:preliminaries}

We consider hypergraphs the hyperedges of which have multiple sources and multiple targets.
A \emph{hypergraph} over a given label alphabet $\Sigma$ is a system $H = (V,E,\sE,\tE,\lE)$ where $V$ is a finite set of \emph{vertices}, $E$ is a finite set of \emph{hyperedges},
$\sE,\tE\colon E \to V^*$ are two functions assigning to each hyperedge a sequence of \emph{sources} and \emph{targets}, respectively,
and $\lE\colon E \to \Sigma$ is a function, called \emph{labeling}.
The components of $H = (V,E,\sE,\tE,\lE)$ may also be denoted by $V_H$, $E_H$, $\sE_H$, $\tE_H$, and $\lE_H$ respectively. The class of all hypergraphs over $\Sigma$ is denoted by $\HSigma$.

Let $pr\colon V^* \times \N \to V$ be defined as $pr(v_1v_2\ldots v_n, i) = v_i$ if $1 \le i \le n$, where $n$ is the length of the sequence. It is undefined otherwise.

Let $H \in \HSigma$, and let $\equiv$ be an equivalence relation on $V_H$.
Then the \emph{fusion of the vertices in $H$ with respect to $\equiv$} yields the hypergraph
$H/{\equiv} = ( V_H/{\equiv}, E_H, s_{H/{\equiv}}, t_{H/{\equiv}}, \lE_H)$
with the set of equivalence classes $V_H/{\equiv} = \{ [v] \mid v \in V_H \}$ and 
$s_{H/{\equiv}}(e) = [v_1]\cdots [v_{k_1}]$, $t_{H/{\equiv}}(e) = [w_{1}]\cdots [w_{k_2}]$ for each $e \in E_H$ with $s_H(e) = v_1\cdots v_{k_1}$, $t_H(e) = w_{1}\cdots w_{k_2}$.
We often use the notation of specifying only the equivalent vertices.

Given $H,H' \in \HSigma$, a \emph{hypergraph morphism} $g \colon H \to H'$ consists of two mappings $g_V\colon V_H \to V_{H'}$ and $g_E\colon E_H\to E_{H'}$ such that $s_{H'}(g_E(e)) = g^*_V(s_H(e))$, $t_{H'}(g_E(e)) = g^*_V(t_H(e))$ and $\lE_{H'}(g_E(e)) = \lE_H(e)$ for all $e \in E_H$, where $g_V^*\colon V_{H}^* \to V_{H'}^*$ is the canonical extension of $g_V$, given by $g^*_V(v_1\cdots v_n) = g_V(v_1)$  $\cdots g_V(v_n)$ for all $v_1\cdots v_n \in V_H^*$.

Given $H,H' \in \HSigma$, $H$ is a \emph{subhypergraph} of $H'$,
denoted by $H \subseteq H'$,
if $V_H \subseteq V_{H'}$, $E_H \subseteq E_{H'}$, $s_H(e) = s_{H'}(e)$, $t_H(e) = t_{H'}(e)$, and $\lE_H(e) = \lE_{H'}(e)$ for all $e \in E_H$.
$H \subseteq H'$ implies that the two inclusions $V_H \subseteq V_{H'}$ and $E_{H} \subseteq E_{H'}$ form a hypergraph morphism from $H \to H'$.

Let $H' \in \HSigma$ as well as $V\subseteq V_{H'}$ and $E\subseteq E_{H'}$. Then the \emph{removal} of $(V,E)$ from $H'$ given by
$H = H' -(V,E) = (V_{H'} - V, E_{H'} - E, s_H, t_H, \lE_H)$
with $s_H(e) = s_{H'}(e)$, $t_H(e) = t_{H'}(e)$ and $\lE_{H}(e) = \lE_{H'}(e)$ for all $e \in E_{H'} - E$
defines a subgraph $H \subseteq H'$ if $s_{H'}(e), t_{H'}(e) \in (V_{H'}-V)^*$ for all $e \in E_{H'} - E$.
Let $H \in \HSigma$, $H' \subseteq H$.
Then $H - H' = H - (V_{H'}, E_{H'})$.

Let $H\in \HSigma$ and $H' = (V',E',s',t'\colon E' \to (V_H + V')^*, \lE'\colon E'\to \Sigma)$ be some quintuple with two sets $V'$, $E'$ and three mappings $s',t'$ and $\lE'$ where $+$ denotes the disjoint union of sets.
Then the \emph{extension} of $H$ by $H'$ given by
$H'' = (V_H + V', E_H + E', s,t, \lE)$
with $s(e) = s_H(e)$, $t(e) = t_H(e)$ and $\lE(e) = \lE_H(e)$ for all $e \in E_H$ as well as $s(e) = s'(e)$, $t(e) = t'(e)$ and $\lE(e) = \lE'(e)$ for all $e \in E'$ is a hypergraph with $H \subseteq H''$.

Let $H \in \HSigma$ and let $\att(e)$ be the set of source and target vertices for $e \in E_H$.
$H$ is \emph{connected} if for each $v,v' \in V_H$, there exists a sequence of triples
$(v_1, e_1, w_1)\ldots(v_n,e_n,w_n) \in (V_H \times E_H \times V_H)^*$ 
such that $v = v_1, v' = w_n$ and $v_i,w_i \in \att(e_i)$ for $i=1,\ldots,n$ and
$w_i = v_{i+1}$ for $i=1,\ldots,n-1$.
A subgraph $C$ of $H$, denoted by $C \subseteq H$, is a \emph{connected component} of $H$ if it is connected and there is no larger connected subgraph, i.e., $C \subseteq C' \subseteq H$ and $C'$ connected implies $C = C'$.
The set of connected components of $H$ is denoted by $\calC(H)$.

Given $H,H' \in \HSigma$, the \emph{disjoint union} of $H$ and $H'$ is denoted by $H+H'$.
Further, $k \cdot H$ denotes the disjoint union of $H$ with itself $k$ times.
We use the \emph{multiplication} of $H$ defined by means of $\calC(H)$ as follows. Let $m \colon \calC(H) \to \N$ be a mapping, called \emph{multiplicity}, then $m\cdot H = \sum_{C \in \calC(H)} m(C) \cdot C$.

A string can be represented by a simple path where the sequence of labels along the path equals the given string.
Let $w= x_1\dots x_n \in \Sigma^*$ for $n\ge 1$ and $x_i \in \Sigma$ for $i=1,\ldots,n$.
Let $[n] = \{1,\ldots,n\}$.
Then the \emph{string graph} of $w$ is defined by $\sg(w) = (\{0\} \cup [n], [n], \sE_w, \tE_w, \lE_w)$
with $\sE_w(i) = (i-1), \tE_w(i) = i$ and $\lE(i) = x_i$ for $i=1,\ldots,n$.
The string graph of the empty string $\varepsilon$, denoted by $\sg(\varepsilon)$, is the discrete graph with a single vertex $0$.
Obviously, there is a one-to-one correspondence between $\Sigma^*$ and $\sg(\Sigma^*) = \{ \sg(w) \mid w \in \Sigma^* \}$.
We define a mapping $begin$ assigning to every string graph to its vertex~$0$.

\section{Turing Machines}
\label{sec:turing-machines}
In this section, we shortly recall the notion of Turing machines (see, e.g., \cite{Turing:36,Hromkovic:04,Hopcroft-Motwani-Ullman:03}) and their recognized languages.
We consider Turing machines with a designated start and accept state and one two-sided infinitely extendable (working) tape. We use two delimiters $\tmtapestart$ and $\tmtapeend$ to indicate the end of the tape to the left and to the right, respectively.
If the head moves beyond a delimiter a new cell labeled $\tmblank$ (the blank symbol) is added.

\begin{definition}
\label{def:TM}
  \begin{enumerate}
  \item A \emph{Turing machine} is a system $\TM = (Q,\Omega, \Gamma, \Delta)$, where
        $Q$ is a finite set of states with two designated different states $q_{start}$ and $q_{accept}$,
        $\Omega$ is the input alphabet,
        $\Gamma$ is the tape alphabet with $\Omega \subseteq \Gamma$ and $\tmblank \in \Gamma \setminus \Omega$,
        and
        $\Delta \subseteq (Q\setminus \{q_{accept}\}) \times \Gamma \times \Gamma \times \{l,n,r\} \times Q$ is the transition relation.

  \item
    $\mathit{conf}(\TM) =  Q \times \Gamma^* \times  \Gamma^*$
    is the set of \emph{configurations}.

  \item A \emph{step} of $\TM$ is defined by the relation $\tmstep{\TM} \subseteq \mathit{conf}(\TM) \times \mathit{conf}(\TM)$:
  \begin{align*}
  (p, \alpha u, x \beta) & \tmstep{\TM} (q, \alpha, u y \beta) & \text{if }& (p,x,y,l,q) \in \Delta\\
  (p, \varepsilon, x \beta) & \tmstep{\TM} (q, \varepsilon, \tmblank y \beta) & \text{if }& (p,x,y,l,q) \in \Delta\\
  (p, \alpha, \varepsilon) & \tmstep{\TM} (q, \alpha, y) & \text{if }& (p,\tmblank,y,l,q) \in \Delta\\
  (p, \alpha, x \beta) & \tmstep{\TM} (q, \alpha, y \beta) & \text{if }& (p,x,y,n,q) \in \Delta\\
  (p, \alpha, x \beta) & \tmstep{\TM} (q, \alpha y, \beta) & \text{if }& (p,x,y,r,q) \in \Delta\\
  (p, \alpha, \varepsilon) & \tmstep{\TM} (q, \alpha y, \varepsilon) & \text{if }& (p,\tmblank,y,r,q) \in \Delta 
    \end{align*}
where $\alpha, \beta \in \Gamma^*, u \in \Gamma$ and $\varepsilon$ is the empty string.

  \item
  A \emph{computation} of $\TM$ is a potentially infinite sequence of configurations $c_0,c_1,\ldots$ where $c_0 = (q_{start} \times \varepsilon \times w)$ is the start configuration wrt the input $w \in \Omega^*$, and $c_i \tmstep{\TM} c_{i+1}$ for all $i \in \N$.
    
  \item The recognized language of $\TM$ is defined as
    $L(\TM) = \{ w  \in \Omega^* \mid (q_{start}, \varepsilon, w) \tmstep{\TM}^* (q_{accept},\alpha,\beta) \},$
    where $\alpha,\beta \in \Gamma^*$ are arbitrary. 
  \end{enumerate}
\end{definition}

\begin{remark}
\begin{enumerate}
\item
$(p,x,y,dir,q) \in \Delta$ means if the Turing machine is in state $p$ and reads the symbol $x$, it can replace $x$ by $y$ and move the (read/write) head to the left if $dir=l$, to the right if $dir=r$ or leave the head stationary if $dir=n$. Afterwards the machine is in state~$q$.

\item A configuration is of the form $(q, \alpha, \beta)$ which means the machine is in state $q$ and the contents of the tape to the left and right of the head are $\alpha$ and $\beta$, respectively. The machine reads the first symbol of $\beta$ if $\beta \ne \varepsilon$ and $\tmblank$ otherwise.

\item
A computation is finite if a halting configuration is reached, i.e., if there is no possibility of continuing the computation.
If the machine enters the state $q_{accept}$, then it accepts the input.

\item The recognized language consists of all strings for which a computation exists such that the machine enters the accepting state $q_{accept}$.

\end{enumerate}
\end{remark}

\section{Context-Dependent Fusion Grammars}
\label{sec:fusion-with-context}
In this section, we recall the notion of context-dependent fusion grammars (cf.~\cite{Kreowski-Kuske-Lye:19a}).
Context-dependent fusion grammars generate hypergraph languages from start hypergraphs via successive applications of context-dependent fusion rules, multiplications of connected components, and a filtering mechanism.
A fusion rule is defined by two complementary-labeled hyperedges and positive and negative context-conditions.
Such a rule is applicable if both the positive and negative context-conditions of the rule are satisfied. 
Its application consumes the two hyperedges and fuses the sources of the one hyperedge with the sources of the other as well as the targets of the one with the targets of the other.

\begin{definition}
\begin{enumerate}
\item 
$F \subseteq \Sigma$ is a \emph{fusion alphabet}
if it is accompanied by a \emph{complementary fusion alphabet}
$\fuscomp{F} = \{ \fuscomp{A} \mid A \in F \} \subseteq \Sigma$,
where $F \cap \fuscomp{F} = \emptyset$ and
$\fuscomp{A} \ne \fuscomp{B}$ for $A,B\in F$ with $A \ne B$
and a \emph{type function}
$type\colon F \cup \fuscomp{F} \rightarrow (\mathbb{N} \times \mathbb{N})$
with
$type(A) = type(\fuscomp{A})$ for each $A \in F$.

\item
For each $A \in F$ with $type(A) = (k_1,k_2)$,  the \emph{fusion rule} $\fr(A)$ is the hypergraph,
depicted in Figure~\ref{fig:fusion rule},
with 
 $V_{\fr(A)} = \{v_i, v'_i \mid i = 1, \ldots , k_1 \}\cup \{w_j, w'_j\mid  j = 1, \ldots ,  k_2\}$,
 $E_{\fr(A)} = \{e,\fuscomp{e}\}$,
 $s_{\fr(A)} (e) = v_1\cdots v_{k_1}$, $s_{\fr(A)} (\fuscomp{e}) = v'_1\cdots v'_{k_1}$,
 $t_{\fr(A)} (e) = w_1\cdots w_{k_2}$, $t_{\fr(A)} (\fuscomp{e}) = w'_1\cdots w'_{k_2}$, and
 $\lE_{\fr(A)}(e) = A$ and $\lE_{\fr(A)}(\fuscomp{e}) = \fuscomp A$.

\item 
The application of $\fr(A)$ to a hypergraph $H \in \HSigma$ proceeds according to the following steps:
(1)~Choose a  \emph{matching morphism} $g \colon \fr(A) \to H$.
(2)~Remove the images of the two hyperedges of $\fr(A)$ yielding $X = H - (\emptyset, \{g(e),g(\fuscomp{e})\})$.
(3)~Fuse the corresponding source and target vertices of the removed hyperedges yielding the hypergraph $ H' = X/{\equiv}$ where $\equiv$ is generated by the
relation $\{(g(v_i),g(v'_i))\mid i = 1,\ldots, k_1\}\cup \{(g(w_j),g(w'_j))\mid j = 1,\ldots,  k_2\}$.
The application of $\fr(A) $ to $H$ is denoted by $H\dder_{\fr(A)} H'$ and called a \emph{direct derivation}.

\begin{figure}[tb]
\centering
\begin{tikzpicture}
  \node [circle,draw=black] (in1a) at (-0.5,1) {} node at (-0.5,1.4) {$v_{k_1}$};
  \node (dots1) at (-1,1) {$\dots$};
  \node [circle,draw=black] (in2a) at (-1.5,1) {} node at (-1.5,1.4) {$v_1$};
  \node [circle,draw=black] (in2b) at (0.5,1) {} node at (0.5,1.4) {$v_{1}'$};
  \node (dots1) at (1,1) {$\dots$};
  \node [circle,draw=black] (in1b) at (1.5,1) {} node at (1.5,1.4) {$v_{k_1}'$};
  
  \node [draw=black] (Aa) at (-1,0) {$A$};
  \node [draw=black] (Ab) at (1,0) {$\fuscomp{A}$};

  \node [circle,draw=black] (out1a) at (-0.5,-1) {} node at (-.5,-1.4) {$w_{k_2}$};
  \node (dots1) at (-1,-1) {$\dots$};
  \node [circle,draw=black] (out2a) at (-1.5,-1) {} node at (-1.5,-1.4) {$w_{1}$};
  \node [circle,draw=black] (out2b) at (0.5, -1) {} node at (.5,-1.4) {$w_{1}'$};
  \node (dots1) at (1,-1) {$\dots$};
  \node [circle,draw=black] (out1b) at (1.5,-1) {} node at (1.5,-1.4) {$w_{k_2}'$};
  
  \foreach \x in {a,b}{
  \path
  (in1\x) edge [->] node[right] {$k_1$} (A\x)
  (in2\x) edge [->] node[left] {$1$} (A\x)
  (A\x) edge [->] node[right] {$k_2$} (out1\x)
  (A\x) edge [->] node[left] {$1$} (out2\x);
  }
\end{tikzpicture}    	
\caption{The fusion rule $\fr(A)$ with $type(A) = (k_1,k_2)$}
\label{fig:fusion rule}
\end{figure}
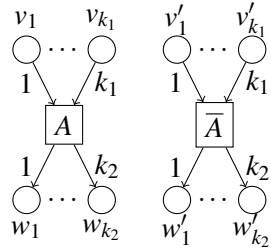

\item A \emph{context-dependent fusion rule} is a triple
$\cdfr = (\fr(A), PC,\allowbreak NC)$ for some $A \in F$ where
$PC$ and $NC$ are two finite sets of hypergraph morphisms with domain $\fr(A)$ mapping into finite contexts defining \emph{positive} and \emph{negative context conditions} respectively.

\item The rule $\cdfr$ is applicable to some hypergraph $H$ via a matching morphism $g\colon \fr(A) \to H$ if
for each $(c\colon \fr(A) \to C) \in PC$ there exists a hypergraph morphism
  $h\colon C \to H$ such that
  $h$ is injective on the set of hyperedges and
  $h \circ c = g$, and
for all $(c\colon \fr(A) \to C) \in NC$ there does not exist a hypergraph morphism
  $h\colon C \to H$ such that
  $h \circ c = g$.
\item If $\cdfr$ is applicable to $H$ via $g$, then the direct derivation $H \dder_{\cdfr} H'$ is the direct derivation $H \dder_{\fr(A)} H'$.
\end{enumerate}
\end{definition}

\begin{remark}
$\fr(A)$ and $(\fr(A),\emptyset,\emptyset)$ are equivalent.
We use the first as an abbreviation for the latter.
\end{remark}

Given a finite hypergraph, the set of all possible successive fusions is finite as fusion rules never create anything. 
To overcome this limitation, arbitrary multiplications of disjoint components within derivations are allowed.
The generated language consists of the terminal part of all resulting connected components that contain no fusion symbols and at least one marker symbol, where marker symbols are removed in the end.
These marker symbols allow us to distinguish between wanted and unwanted terminal components.

\begin{definition}
\begin{enumerate}
\item
A \emph{context-dependent fusion grammar} is a system $\CDFG = (Z,F,M,T,P)$
where $Z \in \mathcal{H}_{F \cup \fuscomp{F} \cup T \cup M}$ is a \emph{start hypergraph} consisting of a finite number of connected components, $F \subseteq \Sigma$ is a finite fusion alphabet, $M \subseteq \Sigma$ with $M \cap (F \cup \fuscomp{F}) = \emptyset$ is a finite set of \emph{markers}, $T \subseteq \Sigma$ with $T \cap (F \cup \fuscomp{F}) = \emptyset = T \cap M$ is a finite set of \emph{terminal labels}, and $P$ is a finite set of context-dependent fusion rules.

\item A \emph{direct derivation} $H \dder H'$ is either a context-dependent fusion rule application $H \dder_{\cdfr} H'$ for some $\cdfr \in P$
or a multiplication $H \dder_{m} m \cdot H$ for some multiplicity~$m\colon \calC(H)\to \N$. 
A \emph{derivation} $H \dder^{n} H'$ of length $n \ge 0$ is a sequence of direct derivations $H_0 \dder H_1 \dder \dots\allowbreak \dder H_n$ with $H=H_0$ and $H' = H_n$.
If the length does not matter, we may write $H \dder^{*} H'$.

\item $L(\CDFG) = \{ rem_M (Y) \mid  Z \dder^{*} H, Y \in \calC(H) \cap (\HTM \setminus \HT) \}$
is the \emph{generated language}
where $rem_M(Y)$ is the terminal hypergraph obtained by removing all hyperedges with labels in $M$ from~$Y$.
\end{enumerate}
\end{definition}

\section{Transformation of Turing Machines into Context-Dependent Fusion\\ Grammars}
\label{sec:trans-tm-to-cdfg}
In this section, we show that Turing machines can be simulated by context-dependent fusion grammars. The construction works roughly as follows:
(1)~A Turing machine is represented by the usual state graph,
(2)~the tape is represented by a sequence of successive edges each labeled with symbols of the working alphabet,
(3)~the state graph and the tape are connected by a hyperedge, called $head$, which also indicates the current state (specifically, the current state is the first source) and is attached to the current position on the tape,
(4)~in addition, the start hypergraph contains components that allow to generate the initial tape in a terminal and a fusion version,
(5)~components that allow to simulate a transition step of the Turing machine by a sequence of applications of context-dependent fusion rules,
and (6)~there is a terminating component that enables to disconnect the terminal tape with the input string from the rest of the working hypergraph whenever an accepting state is reached. In  other words, the grammar generates a tape with a detachable input string if and only if the Turing machine accepts this string.

Because the transformation is quite complicated we introduce the ideas step by step. 
First, we give a hypergraph representation of Turing machines.
Then we introduce the tape graph representing the working tape as well as the input to the Turing machine.
This leads us directly to hypergraph representations of configurations. 
Afterwards, we demonstrate how a step can be simulated by a sequence of context-dependent fusion rules.
Finally, the two constructions are combined and our main theorem is presented.

\subsection{Representation of a Turing machine by a hypergraph}
      
In the hypergraphical representation of a Turing machine, denoted by $hg(\TM)$, vertices represent states and edges between these vertices represent the elements of the transition relation.
Initially, there is one additional vertex and three special hyperedges:
an $acc$-loop indicates the accepting state,
a hyperedge with $|Q|$ sources and one target, called $head$, connects the state graph with
the additional vertex to which a $\fuscomp{tape}$-hyperedge is attached. The latter enables fusion of the Turing machine with the tape graph.

\begin{definition}\label{def:construction-hg_tm}
Let $TM = (Q,\Omega, \Gamma, \Delta)$ be a Turing machine.
Let $\sigma = q_1 \cdots q_{|Q|}$ be a sequence of states, where each state occurs exactly once.
Define
$hg(\TM,\sigma) = (Q + \{v_{head}\},\{acc,head,\fuscomp{tape}\} + \Delta,\sE,\tE,\lE)$,
  where
  $\sE(acc) = \tE(acc) = q_{accept}$, $\lE(acc) = acc$,
  $\sE(head) = \sigma$, $\tE(head) = v_{head}$, $\lE(head) = head$,
  $\sE(\fuscomp{tape}) = v_{head}$, $\tE(\fuscomp{tape}) = \varepsilon$, $\lE(\fuscomp{tape}) = \fuscomp{tape}$,
  $\sE( \delta ) = p$, $\tE( \delta ) = q$, and $\lE( \delta ) = x/y/dir$,
    where $\delta = (p,x,y,dir,q) \in \Delta$.
\end{definition}

\begin{example}\label{example:cdfg-TM-example}
Consider the Turing machine in Figure~\ref{fig:cdfg-TM-example}.
The corresponding hypergraph is depicted in Figure~\ref{fig:cdfg-TM-example-corresponding-hypergraph}
where the $head$-hyperedge is dashed.
\end{example}

\begin{figure}[t]
\begin{subfigure}[b]{0.45\textwidth}
\begin{tikzpicture}[->,>=stealth',shorten >=1pt,auto,node distance=2.8cm,
                    semithick]
  \tikzstyle{every state}=[draw=black,text=black]

  \node at (0,0)(1) {};
  
  \node[state,accepting]   (acc)  [below = 2.05 of 1]      {$q_{accept}$}; 
  \node[state] [right of=acc] (start)                    {$q_{start}$};
  \node[state]         (q1) [right of=start] {$q_{aux}$};
  
  \path (start) edge[bend left=8]              node {$b/\tmblank/r$} (q1);
  \path (q1) edge[bend left=8]              node {$b/b/n$} (start);
  \path (start) edge              node {$a/c/r$} (acc);

\end{tikzpicture}
\caption{the usual state graph of a Turing machine}
\label{fig:cdfg-TM-example}
\end{subfigure}
\begin{subfigure}[b]{0.5\textwidth}
\begin{tikzpicture}[->,>=stealth',shorten >=1pt,auto,node distance=2.8cm,
                    semithick]
  \tikzstyle{every state}=[draw=black,text=black]

  \node [inner sep = 1.5pt,circle,fill=black] at (1,0)(1) {};
  
  \node [transition] at (2,0.5)  (t2)  {$~\fuscomp{tape}~$}
  edge [pre]                    (1);
 
  \node[state,accepting]   (acc)  [below = 1.4 of 1]      {$q_{accept}$};

  \path (acc) edge[loop left]  node[below = 0.1] {$acc$} (acc);

  \node[state] [right of=acc] (start)                    {$q_{start}$};
  \node[state]         (q1) [right of=start] {$q_{aux}$};

  \node [transition] [below = 0.3 of 1]  (head)  {$head$}
  edge [pre,dashed] node [right = 0.01, very near start] {$1$}                    (start)
  edge [pre,bend left=28,dashed] node [below = 0.01, very near start] {$2$}                  (q1)
  edge [pre,dashed] node [below = 0.01, very near start] {$3~~~$}                  (acc)
  edge [post,dashed]                    (1);

  \path (start) edge[bend left=8]              node {$b/\tmblank/r$} (q1);
  \path (q1) edge[bend left=8]              node {$b/b/n$} (start);
  \path (start) edge              node {$a/c/r$} (acc);

\end{tikzpicture}
\caption{the hypergraph representation of the Turing machine}
\label{fig:cdfg-TM-example-corresponding-hypergraph}
\end{subfigure}
  \caption{Example state graph and the corresponding hypergraph representation of a Turing machine}
  \label{Fig:cdfg-TM-example+input}
\end{figure}
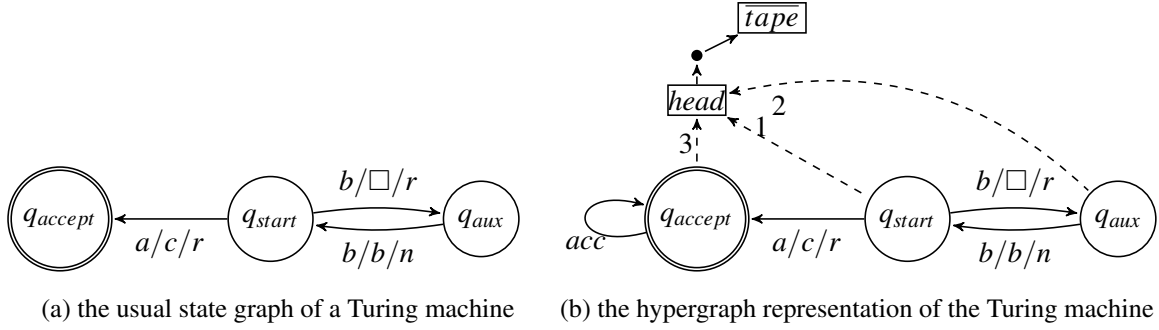

\begin{remark}
  The order in which the states of the Turing machine are connected to the sources of the $head$-hyperedge implements a permutation.
$\sigma$ is permuted when a transition step is simulated.
\end{remark}

\subsection{Tape graph}

In our construction the tape is represented by an infinitely extendable tape graph.
Due to technical reasons, the tape graph contains two connected string graphs, where one
is labeled over the terminal alphabet $\Omega$ and the other is labeled over the fusion alphabet $\fsym{\Gamma}$ (defined later).
The construction can be seen as having two tapes (input and working tape) initially with the same content, where the input tape is left invariant, but the working tape may be modified or extended by $\tmblank$-cells.
If the machine halts in the accepting state, then the content of the input tape is used as a contribution to the generated language.

In our construction five additional hyperedges are used.
The terminal-labeled string graph carries a marker hyperedge. The two corresponding string graphs are connected via a hyperedge labeled $cut$ which is used to disconnect the terminal- and marker-labeled string graph.
A $tape$-hyperedge is connected to the first vertex of the fusion-labeled string graph and is later used for attaching the tape graph to a hypergraph representation of a Turing machine.
Two hyperedges labeled $\tmtapestart$ and $\tmtapeend$ are used to extend the fusion-symbol labeled string graph with $\tmblank$-labeled hyperedges an unbounded number of times via the connected components
\begin{tikzpicture}[baseline=-0.1cm,node distance=1.3cm,>=stealth',bend angle=25,auto]

  \node [inner sep = 1.5pt,circle,fill=black] at (0,0)(0) {};
  \node [inner sep = 1.5pt,circle,fill=black] at (1.5,0)(1) {};

  \node [transition] at (0.75,0)  (blank)  {$\tmblank$}
  edge [pre]                     (0)
  edge [post]                     (1);

  \node [transition] at (-0.7,0) (start)  {$\tmtapestart$}
  edge [post]                     (0);

  \node [transition] at (2.25,0) (start)  {$\fuscomp{\tmtapestart}$}
  edge [post]                     (1);

\end{tikzpicture}
and
\begin{tikzpicture}[baseline=-0.1cm,node distance=1.3cm,>=stealth',bend angle=25,auto]

  \node [inner sep = 1.5pt,circle,fill=black] at (0,0)(0) {};
  \node [inner sep = 1.5pt,circle,fill=black] at (1.5,0)(1) {};

  \node [transition] at (0.7,0)  (blank)  {$\tmblank$}
  edge [pre]                     (0)
  edge [post]                     (1);

  \node [transition] at (2.25,0) (end)  {$\tmtapeend$}
  edge [pre]                     (1);

  \node [transition] at (-0.7,0) (end)  {$\fuscomp{\tmtapeend}$}
  edge [pre]                     (0);

\end{tikzpicture}
as well as the fusion rules $\fr(\tmtapestart)$ and~$\fr(\tmtapeend)$.

$\fsym{\Gamma}= (\Gamma \setminus \Omega) + \fsym{\Omega}$, where $\fsym{\Omega} = \{ \fsym{x} \mid x \in \Omega\}$. This is because in fusion grammars fusion alphabets and terminal alphabets are disjoint but $\Omega \subsetneq \Gamma$ by definition of the Turing machine.
For example, the string $ab$ is represented by the graphs $\sg(ab)$ and $\encsg{(f(ab))} = $
\begin{tikzpicture}[baseline=-1mm]
  \foreach \x in {2,4,6}{
     \node [inner sep = 1.5pt,circle,fill=black] at (\x,0)(\x) {};
  }
  \foreach \x/\z [evaluate=\x as \y using \x+2,evaluate=\x as \p using \x+1] in {2/\fsym{a},4/\fsym{b}}{    
    \node [transition] at (\p,0)  (t\x)  {$\z$}
      edge [pre]                    (\x)
      edge [post,shorten >= 5pt]                   (\y);
  }
  \node [transition] at (1,0)  (tapestart)  {$\tmtapestart$}
      edge [post]                    (2);

  \node [transition] at (7,0)  (tapeend)  {$\tmtapeend$}
      edge [pre]                    (6);
\end{tikzpicture}
, where $f\colon \Gamma^* \to \fsym{\Gamma}^*$ is defined by $f(x) = \fsym{x}$, if $x \in \Omega$, and $f(x) = x$, otherwise.

\begin{figure}[t]
\centering
\begin{subfigure}[t]{0.68\textwidth}
\begin{tikzpicture}[baseline=0.0cm,node distance=1.3cm,>=stealth',bend angle=45,auto]
  \node [inner sep = 1.5pt,circle,fill=black] at (1,0)(0) {}; 
  \node [inner sep = 1.5pt,circle,fill=black] at (1,1.25)(1) {}; 
 
  \node [inner sep = 1.5pt,circle,fill=black] at (2.5,0)(2) {};
  \node [inner sep = 1.5pt,circle,fill=black] at (2.5,1.25)(3) {};

  \node [inner sep = 1.5pt,circle,fill=black] at (4,0)(4) {};
  \node [inner sep = 1.5pt,circle,fill=black] at (4,1.25)(5) {};

  \node [inner sep = 1.5pt,circle,fill=black] at (5.5,0)(6) {};
  \node [inner sep = 1.5pt,circle,fill=black] at (5.5,1.25)(7) {};

  \node [inner sep = 1.5pt,circle,fill=black] at (-0.5,0)(8) {};
  \node [inner sep = 1.5pt,circle,fill=black] at (-2,0)(9) {};
  \node [inner sep = 1.5pt,circle,fill=black] at (-3.5,0)(10) {};
  
  \node [transition] at (1.5,-0.5)  (t1)  {$tape$}
  edge [post]                    (0);

  \node [transition] at (0.45,0.6)  (t2)  {$cut$}
  edge [pre]                    (0)
  edge [post]                   (1);
  
  \path (1) edge [loop left] node{$\mu$}();

  \node [transition] at (1.75,0)  (t3)  {$\fsym{{\beta_1}}$}
  edge [pre]                    (0)
  edge [post]                   (2);
  \path (1) edge[->] node[above=.06]{$w_1$} (3);

  \node [transition] at (4.75,0)  (t3)  {$\fsym{{\beta_m}}$}
  edge [pre]                    (4)
  edge [post]                   (6);
  \path (5) edge[->] node[above=.06]{$w_k$} (7);

  \node [transition] at (-2.75,0)  (t3)  {$\fsym{{\alpha_1}}$}
  edge [pre]                    (10)
  edge [post]                   (9);
  \node [transition] at (0.25,0)  (t3)  {$\fsym{{\alpha_n}}$}
  edge [pre]                    (8)
  edge [post]                   (0);
    
  \path (2) edge[dotted] (4);
  \path (3) edge[dotted] (5);
  \path (8) edge[dotted] (9);
  
  \node [transition] at (-4,0)  (t3)  {$\tmtapestart$}
  edge [post]                   (10);
  
  \node [transition] at (6.3,0)  (t3)  {$\tmtapeend$}
  edge [pre]                   (6);

\end{tikzpicture}
\caption{Schematic drawing of $\tg(\alpha,\beta,w)_{tape}$, where $\alpha = \alpha_1\cdots \alpha_n$, $\beta = \beta_1\cdots \beta_m$, and $w = w_1\cdots w_k$ }
\label{Fig:cdfg-schema-tape-graph}
\end{subfigure}
\begin{subfigure}[t]{0.3\textwidth}
\begin{tikzpicture}[baseline=0.0cm,node distance=1.3cm,>=stealth',bend angle=45,auto]
  \node [inner sep = 1.5pt,circle,fill=black] at (1,0)(0) {};    
  \node [inner sep = 1.5pt,circle,fill=black] at (1,1.25)(1) {};

  \node [inner sep = 1.5pt,circle,fill=black] at (2.5,0)(2) {};
  \node [inner sep = 1.5pt,circle,fill=black] at (2.5,1.25)(3) {};

  \node [inner sep = 1.5pt,circle,fill=black] at (4,0)(4) {};
  \node [inner sep = 1.5pt,circle,fill=black] at (4,1.25)(5) {};

  \node [transition] at (1.5,-0.5)  (t1)  {$tape$}
  edge [post]                    (0);

  \node [transition] at (0.45,0.6)  (t2)  {$cut$}
  edge [pre]                    (0)
  edge [post]                   (1);
  \node [transition] at (0.34,0)  (t3)  {$\tmtapestart$}
  edge [post]                   (0);
  
  \path (1) edge [loop left] node{$\mu$}();

  \node [transition] at (1.75,0)  (t3)  {$\fsym{a}$}
  edge [pre]                    (0)
  edge [post]                   (2);
  \path (1) edge[->] node[above=.06]{$a$} (3);
  \node [transition] at (3.25,0)  (t3)  {$\fsym{b}$}
  edge [pre]                    (2)
  edge [post]                   (4);
  \path (3) edge[->] node[above=.06]{$b$} (5);

  \node [transition] at (4.8,0)  (t3)  {$\tmtapeend$}
  edge [pre]                   (4);

\end{tikzpicture}
\caption{$\tg(\varepsilon, ab,ab)_{tape}$}
\label{Fig:cdfg-TM-tape-graph-ab}
\end{subfigure}

\begin{subfigure}[b]{0.15\textwidth}
\begin{tikzpicture}[baseline=0.0cm,node distance=1.3cm,>=stealth',bend angle=45,auto]
  \node [inner sep = 1.5pt,circle,fill=black] at (1,0)(0) {};    
  \node [inner sep = 1.5pt,circle,fill=black] at (1,1.25)(1) {};
  
  \node [transition] at (1.5,-0.5)  (t1)  {$tape$}
  edge [post]                    (0);
  \node [transition] at (1.55,0.6)  (t2)  {$gen$}
  edge [pre]                    (0)
  edge [post]                   (1);
  \node [transition] at (0.45,0.6)  (t2)  {$cut$}
  edge [pre]                    (0)
  edge [post]                   (1);
  \node [transition] at (0.34,0)  (t3)  {$\tmtapestart$}
  edge [post]                   (0);
  
  \path (1) edge [loop left] node{$\mu$}();
\end{tikzpicture}
\caption{$tape_{start}$}
\label{fig:cdfg-tm-wstart}
\end{subfigure}
~~~~~
\begin{subfigure}[b]{0.18\textwidth}
\begin{tikzpicture}[baseline=0.0cm,node distance=1.3cm,>=stealth',bend angle=45,auto]
  \node [inner sep = 1.5pt,circle,fill=black] at (0,0)(0) {};    
  \node [inner sep = 1.5pt,circle,fill=black] at (2,0)(1) {};
  \node [inner sep = 1.5pt,circle,fill=black] at (0,1.25)(2) {};
  \node [inner sep = 1.5pt,circle,fill=black] at (2,1.25)(3) {};

  \node [transition] at (1,0)  (t1)  {$\fsym{x}$}
  edge [pre]                    (0)
  edge [post]                   (1);
  \node [transition] at (2,0.6)  (t2)  {$gen$}
  edge [pre]                    (1)
  edge [post]                   (3);
  \node [transition] at (0,0.6)  (t3)  {$\fuscomp{gen}$}
  edge [pre]                    (0)
  edge [post]                   (2);

  \path (2) edge[->] node[above=.06]{$x$} (3);
  
\end{tikzpicture}
\caption{$tape_{x}$}
\label{fig:cdfg-tm-wx}
\end{subfigure}
~~~~~
\begin{subfigure}[b]{0.1\textwidth}
\begin{tikzpicture}[baseline=0.0cm,node distance=1.3cm,>=stealth',bend angle=45,auto]

  \node [inner sep = 1.5pt,circle,fill=black] at (0,0)(0) {};    
  \node [inner sep = 1.5pt,circle,fill=black] at (0,1.25)(1) {};

  \node [transition] at (0,0.6)  (t1)  {$~\fuscomp{gen}~$}
  edge [pre]                    (0)
  edge [post]                   (1);
  \node [transition] at (0.7,0)  (t3)  {$\tmtapeend$}
  edge [pre]                    (0);
  
\end{tikzpicture}
\caption{$tape_{end}$}
\label{fig:cdfg-tm-wend}
\end{subfigure}
~~~~~
\begin{subfigure}[b]{0.22\textwidth}
\begin{tikzpicture}[baseline=-0.1cm,node distance=1.3cm,>=stealth',bend angle=25,auto]

  \node [inner sep = 1.5pt,circle,fill=black] at (0,0)(0) {};
  \node [inner sep = 1.5pt,circle,fill=black] at (1.5,0)(1) {};

  \node [transition] at (0.75,0)  (blank)  {$\tmblank$}
  edge [pre]                     (0)
  edge [post]                     (1);

  \node [transition] at (-0.7,0) (start)  {$\tmtapestart$}
  edge [post]                     (0);

  \node [transition] at (2.25,0) (start)  {$\fuscomp{\tmtapestart}$}
  edge [post]                     (1);

\end{tikzpicture}

~

\begin{tikzpicture}[baseline=-0.1cm,node distance=1.3cm,>=stealth',bend angle=25,auto]

  \node [inner sep = 1.5pt,circle,fill=black] at (0,0)(0) {};
  \node [inner sep = 1.5pt,circle,fill=black] at (1.5,0)(1) {};

  \node [transition] at (0.7,0)  (blank)  {$\tmblank$}
  edge [pre]                     (0)
  edge [post]                     (1);

  \node [transition] at (2.25,0) (end)  {$\tmtapeend$}
  edge [pre]                     (1);

  \node [transition] at (-0.7,0) (end)  {$\fuscomp{\tmtapeend}$}
  edge [pre]                     (0);

\end{tikzpicture}
\caption{$tape_{\tmtapestart}$ and $tape_{\tmtapeend}$}
\label{fig:cdfg-tm-wext}
\end{subfigure}

\begin{subfigure}[b]{0.6\textwidth}
  \begin{tikzpicture}[baseline=0.0cm,node distance=1.3cm,>=stealth',bend angle=45,auto]
  \node [inner sep = 1.5pt,circle,fill=black] at (1,0)(0) {};    
  \node [inner sep = 1.5pt,circle,fill=black] at (1,1.25)(1) {};
  
  \node [transition] at (1.5,-0.5)  (t1)  {$tape$}
  edge [post]                    (0);
  \node [transition] at (1.55,0.6)  (t2)  {$gen$}
  edge [pre]                    (0)
  edge [post]                   (1);
  \node [transition] at (0.45,0.6)  (t2)  {$cut$}
  edge [pre]                    (0)
  edge [post]                   (1);
  \node [transition] at (0.34,0)  (t3)  {$\tmtapestart$}
  edge [post]                   (0);
  
  \path (1) edge [loop left] node{$\mu$}();
\end{tikzpicture}
  \newcommand{\cdfgtmwparam}[1]{
\begin{tikzpicture}[baseline=0.0cm,node distance=1.3cm,>=stealth',bend angle=45,auto]
  \node [inner sep = 1.5pt,circle,fill=black] at (0,0)(0) {};    
  \node [inner sep = 1.5pt,circle,fill=black] at (2,0)(1) {};
  \node [inner sep = 1.5pt,circle,fill=black] at (0,1.25)(2) {};
  \node [inner sep = 1.5pt,circle,fill=black] at (2,1.25)(3) {};

  \node [transition] at (1,0)  (t1)  {$\fsym{#1}$}
  edge [pre]                    (0)
  edge [post]                   (1);
  \node [transition] at (2,0.6)  (t2)  {$gen$}
  edge [pre]                    (1)
  edge [post]                   (3);
  \node [transition] at (0,0.6)  (t3)  {$\fuscomp{gen}$}
  edge [pre]                    (0)
  edge [post]                   (2);

  \path (2) edge[->] node[above=.06]{$#1$} (3);
  
\end{tikzpicture}
}
  \cdfgtmwparam{a}
  \cdfgtmwparam{b}
  \begin{tikzpicture}[baseline=0.0cm,node distance=1.3cm,>=stealth',bend angle=45,auto]

  \node [inner sep = 1.5pt,circle,fill=black] at (0,0)(0) {};    
  \node [inner sep = 1.5pt,circle,fill=black] at (0,1.25)(1) {};

  \node [transition] at (0,0.6)  (t1)  {$~\fuscomp{gen}~$}
  edge [pre]                    (0)
  edge [post]                   (1);
  \node [transition] at (0.7,0)  (t3)  {$\tmtapeend$}
  edge [pre]                    (0);
  
\end{tikzpicture}
\caption{The connected components needed to generate $\tg(\varepsilon, ab,ab)_{tape}$}
\label{Fig:cdfg-TM-tape-graph-ab-components}  
\end{subfigure}
\begin{subfigure}[b]{0.3\textwidth}
  \centering
  \begin{tikzpicture}[->,>=stealth',
                    semithick]
  \tikzstyle{every state}=[draw=black,text=black]

  \node [inner sep = 1.5pt,circle,fill=black] at (0,0)(0a) {};
  \node [inner sep = 1.5pt,circle,fill=black] at (1,0)(1a) {};
  \node [inner sep = 1.5pt,circle,fill=black] at (0.5,1)(2a) {};

  \node [inner sep = 1.5pt,circle,fill=black] at (1.5,0)(0b) {};
  \node [inner sep = 1.5pt,circle,fill=black] at (3,0)(1b) {};
  \node [inner sep = 1.5pt,circle,fill=black] at (2.25,1)(2b) {};
  
  \path (0a) edge[->] node[above=.06]{$a$} (1a);
  \path (1a) edge[->] node[midway, right]{$c$} (2a);
  \path (2a) edge[->] node[near start, left]{$b$} (0a);
    
  \node [transition] at (2.25,0)  (t)  {$\fsym{a}$}
  edge [pre]                    (0b)
  edge [post]                   (1b);

  \node [transition] at (2.6,0.5)  (t)  {$\fsym{c}$}
  edge [pre]                    (1b)
  edge [post]                   (2b);

  \node [transition] at (1.8,0.5)  (t)  {$\fsym{b}$}
  edge [pre]                    (2b)
  edge [post]                   (0b);
  
\end{tikzpicture}
  \caption{hypergraph obtained from $tape_a, tape_c, tape_b$}
\label{fig:cdfg-tm-tape-two-triangle}
\end{subfigure}
\caption{Tape graphs and connected components needed for their generation}
\end{figure}
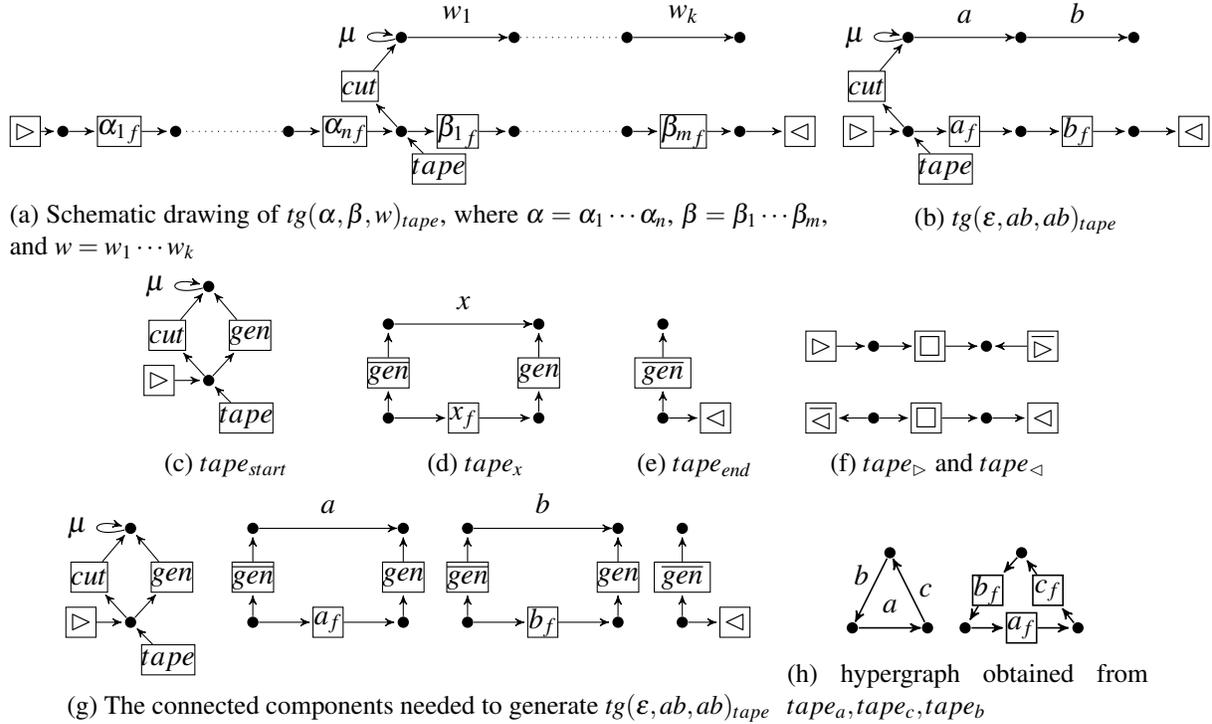

\begin{definition}\label{def:tape-graph}
Let $\alpha,\beta \in \Gamma^*, w\in \Omega^*$.
Let $cut, tape, \tmtapestart$ and $\tmtapeend$ be fusion symbols with
$type(\tmtapestart) = (0,1)$,
$type(\tmtapeend) = type(tape) = (1,0)$, and
$type(cut) = (1,1)$,
and let $\sg(cut \cdot w)_{\mu,tape}$ be the string graph $\sg(cut \cdot w)$ with an additional $\mu$-labeled loop attached to $begin(\sg(w))$ and a $tape$-hyperedge attached to $begin(\sg(cut))$.
Then 
$\tg(\alpha,\beta,w)_{tape} = (\encsg{(f(\alpha \beta))} + \sg(cut \cdot w)_{\mu,tape})/_{ {begin(\sg(f(\beta)))} \equiv \sE(cut) }$
is a \emph{tape graph}.
A schematic drawing is depicted in Figure~\ref{Fig:cdfg-schema-tape-graph}.
\end{definition}

We depict hyperedges with one source and one target with labels $x \in \Sigma \setminus F$ by
\begin{tikzpicture}[inner sep = 1.5pt, baseline = -3pt,cap=round]
\node [circle,fill=black] (0) at (0,0) {};
\node [circle,fill=black] (1) at (1,0) {};
\path (0) edge[->] node[below=.06]{\footnotesize{$x$}} (1);
\end{tikzpicture}.

\begin{example}
The tape graph $\tg(\varepsilon, ab,ab)_{tape}$ is depicted in Figure~\ref{Fig:cdfg-TM-tape-graph-ab}.
\end{example}

Because the input to the Turing machine may be any $w\in \Omega^*$ we need a construction to generate arbitrary tape graphs corresponding to inputs.
This is realized by the following context-dependent fusion grammar, where the fusion rule $\fr(gen)$ is used to fuse $gen$- and $\fuscomp{gen}$-hyperedges in order to generate two corresponding string graphs (one terminal, one fusion labeled), and the fusion rules $\fr(\tmtapestart)$ and $\fr(\tmtapeend)$ are used to add $\tmblank$-edges to the left and right of the fusion-labeled string graph as described above.

\begin{definition}\label{def:construction-CDFG-tg}
Let $cut, \tmtapestart$ and $\tmtapeend$ be as before.
Let $F_{\tg} = \{tape,gen,cut,\tmtapestart,\tmtapeend\} + \fsym{\Gamma}$ be a fusion alphabet
with $type(gen) = type(x) = (1,1)$ for each $x \in \fsym{\Gamma}$.
Define the context-dependent fusion grammar
$\CDFG_{\tg}(\Omega, \Gamma) = (Z_{\tg}, F_{\tg}, \{\mu\}, \Omega, \{ \fr(gen), \fr(\tmtapestart), \fr(\tmtapeend) \}),$
where the start hypergraph $Z_{\tg} =  tape_{start} + tape_{end} + \sum\limits_{x \in \Omega} tape_{x} + tape_{\tmtapestart} + tape_{\tmtapeend}$ consists of the connected components
depicted in Figure~\ref{fig:cdfg-tm-wstart}--\ref{fig:cdfg-tm-wext}.
\end{definition}

\begin{example}\label{example:trans(PN)-3}
Let $\Omega = \{a,b,c\}$ and $\Gamma = \{ a,b,c,\tmblank \}$.
Then
$\CDFG_{tg}(\Omega, \Gamma) = (Z_{example},\{tape,gen,cut,\tmtapestart,\tmtapeend, \fsym{a},\fsym{b},\fsym{c}, \tmblank\}, \{\mu\}, \Omega, \{ \fr(gen), \fr(\tmtapestart), \fr(\tmtapeend) \})$
with
$Z_{example} =  tape_{start} + tape_{end} + tape_{a} + tape_{b} + tape_{c} + tape_{\tmtapestart} + tape_{\tmtapeend}$.

The tape graph $\tg(\varepsilon, ab,ab)_{tape}$, depicted in Figure~\ref{Fig:cdfg-TM-tape-graph-ab},
can be generated by applying $\fr(gen)$ three times to the connected components $tape_{start}, tape_a, tape_b$ and $tape_{end}$,
depicted in Figure~\ref{Fig:cdfg-TM-tape-graph-ab-components}.
However, due to the context-freeness of $\fr(gen)$ fusions within some connected component are also possible yielding e.g. the hypergraph in Figure~\ref{fig:cdfg-tm-tape-two-triangle} obtained from $tape_a, tape_c, tape_b$.
Note that the left connected component is terminal labeled.
However, it does not contribute to the generated language because it lacks a marker hyperedge. 
\end{example}

The following propositions show that this context-dependent fusion grammar generates certain tape graphs. But everything derivable does not contribute to the generated language because there is no possibility to obtain a connected component which is only terminal labeled.
However, with a slight modification of the grammar the generated language is $\Omega^*$ up to representation of strings as graphs.

\begin{proposition}\label{prop:cdfg-tape-graph-1}
For each $i,j \in \N, w\in \Omega^*$
exists a derivation $Z \dder^* tg(\tmblank^i, w\cdot \tmblank^j, w)_{tape}$ in
$\CDFG_{\tg}(\Omega, \Gamma)$.
\end{proposition}

\begin{proof}
  by induction on $i,j$ and the length of $w$.
  
  We first prove
  for each $w = w_1 \ldots w_n \in \Omega^*$ exists a derivation $tape_{start} + tape_{end} + \sum_{x \in \Omega} tape_{x} \dder^{n}\linebreak tg(\varepsilon, w, w)_{tape}$. Therefore, let $\#\colon \Omega^* \times \Omega \to \N$ be a mapping of a string and a symbol to the number of occurrences of this symbol in the string.

Induction base: $n=0$:
Let $m$ be a multiplicity with
$m(tape_{start}) = m(tape_{end}) = 1$,
$m(tape_{x}) = 0$ for each $x \in \Omega$.
Then
$tape_{start} + tape_{end} + \sum_{x \in \Omega} tape_{x} \dder_m tape_{start} + tape_{end} \dder_{\fr(gen)} tg(\varepsilon, \varepsilon, \varepsilon)_{tape}$.

Induction step: $w = w_1 \ldots w_{n+1}, n\ge 0$.
Let
$m(tape_{start}) = m(tape_{end}) = 1$,
$m(tape_x) = \#(w, x)$ for each $x \in \Omega$.
Then by induction hypothesis there is a derivation
$tape_{start} + tape_{end} + \sum_{x \in \Omega} tape_{x} \dder_m\linebreak tape_{start} + tape_{end} + \sum\limits_{x \in \Omega} m(tape_{x})\cdot tape_{x} \dder_{\fr(gen)}^n tg(\varepsilon, w_1 \ldots w_n, w_1\ldots w_n)_{tape} + tape_{w_{n+1}}$.
Due to the context-freeness of $\fr(gen)$ one may assume w.l.o.g. that the connected components involved in the last derivation step are
$X + tape_{end}$
where $X$ is $tg(\varepsilon, w_1 \ldots w_{n}, w_1\ldots w_{n})_{tape}$ without the $\tmtapeend$-hyperedge but with an additional $\fuscomp{gen}$-hyperedge attached to the ends of the string graphs.
Then $X + tape_{w_{n+1}} + tape_{end} \dder^2\linebreak tg(\varepsilon, w_1 \ldots w_{n+1}, w_1\ldots w_{n+1})_{tape}$.

  Now we prove
  $tg(\alpha, \beta, w)_{tape} + tape_{\tmtapestart} \dder^{i+1} tg(\tmblank^i \alpha, \beta,w)_{tape}$ for any $tg(\alpha, \beta, w)_{tape}$.
  The proof for $tg(\alpha, \beta, w)_{tape} + tape_{\tmtapeend} \dder^{j+1} tg( \alpha, \beta \tmblank^j,w)_{tape}$ for any $tg(\alpha, \beta, w)_{tape}$ is analog.

  Induction base: $i=0$.
  Let $m$ be a multiplicity with
  $m(tg(\alpha, \beta, w)_{tape}) = 0$ and
  $m(tape_{\tmtapestart}) = 1$.
  Then
  $tg(\alpha, \beta, w)_{tape} + tape_{\tmtapestart} \dder_m tg(\tmblank^0 \alpha, \beta,w)_{tape}$.

  Induction step:
  Let $m$ be a multiplicity with
  $m(tg(\alpha, \beta, w)_{tape}) = 1$ and
  $m(tape_{\tmtapestart}) = i+1$.
  Then
  $tg(\alpha, \beta, w)_{tape} + tape_{\tmtapestart} \dder_m tg(\alpha, \beta,w)_{tape} + (i+1)\cdot tape_{\tmtapestart} \dder_{\fr(\tmtapestart)}^i tg(\tmblank^i \alpha, \beta, w)_{tape} + tape_{\tmtapestart} \dder_{\fr(\tmtapestart)}\linebreak tg(\tmblank^{i+1} \alpha, \beta, w)_{tape}$.
\end{proof}

\begin{proposition}\label{prop:cdfg-tape-graph-2}
$L(\CDFG_{\tg}(\Omega, \Gamma)) = \emptyset$.
\end{proposition}

\begin{proof}
It is sufficient to focus on the connected components with some marker $\mu$.
  The statement holds because there is no possibility to fuse the $cut$-hyperedge which is attached to the $\mu$-hyperedge. Therefore, no connected component is only terminal and marker labeled.
\end{proof}

\begin{proposition}\label{prop:cdfg-tape-graph-3}
Let
$CDFG_{\tg+\fuscomp{cut}}(\Omega, \Gamma) = (Z_{\tg} + z_{\fuscomp{cut}}, F_{\tg}, \{\mu\}, \Omega, P'_{\tg})$, where
$z_{\fuscomp{cut}}=$
\begin{tikzpicture}[inner sep = 0mm, baseline = -3pt]
  \node [inner sep = 1.5pt,circle,fill=black] at (0,0)(0) {};    
  \node [inner sep = 1.5pt,circle,fill=black] at (1.5,0)(1) {};

  \node [transition,minimum size=4mm] at (0.75,0)  (t2)  {$~\fuscomp{cut}~$}
  edge [pre]                    (0)
  edge [post]                   (1);
\end{tikzpicture}
and
$P'_{\tg} = P_{\tg} \cup \{\fr(cut)\}$.
Then $L(CDFG_{\tg+\fuscomp{cut}}(\Omega, \Gamma)) = \{ \sg(w) \mid w \in \Omega^*\}$.
\end{proposition}

\begin{proof}
$Z \dder^* tg(\tmblank^i, w\cdot \tmblank^j, w)_{tape} + z_{\fuscomp{cut}} \dder_{\fr(cut)}\allowbreak  \encsg{(\tmblank^if(w)\tmblank^j)} + \sg(w)_\mu$
for any $i,j \in \N, w \in \Omega^*$, where
the first part uses the same argument as in Proposition~\ref{prop:cdfg-tape-graph-1} and
$\fr(cut)$ matches the $cut$-hyperedge in the tape graph and the $\fuscomp{cut}$-hyperedge in $z_{\fuscomp{cut}}$.
Because only the latter connected component contains a marker and is marker and terminal labeled it contributes to the language.

The converse follows from the fact that the derivation constructed above is a normal form because the derivation steps are interchangeable\footnote{In \cite{Kreowski-Kuske-Lye:17a} it is shown that every derivation can be rearranged such that first all multiplications and afterwards all (context-free) fusions in arbitrary order are performed.} due to the context-freeness.
\end{proof}

\subsection{Hypergraph representation of a configuration}

A hypergraph representation of a configuration consists of the hypergraph representation of the Turing machine fused to some tape graph. In this way the configuration is interlinked with a specific input and with some permutation of states.
An initial configuration is obtained by fusing a $tape$-hyperedge in a hypergraph representation of the Turing machine with some $\fuscomp{tape}$-hyperedge in a tape graph.

\begin{example} The application of the fusion rule $\fr(tape)$ to the hypergraph representation of the Turing machine in Example~\ref{example:cdfg-TM-example} and the tape graph generated in Example~\ref{example:trans(PN)-3} yields the hypergraph representation of the initial configuration $(q_{start}, \varepsilon, ab)$ wrt the input adjunct $ab$ and the permutation $q_{start}q_{aux}q_{accept}$ depicted on the right-hand site in Figure~\ref{fig:cdfg-TM-example-configureation-tape-fusion}.
\end{example}
  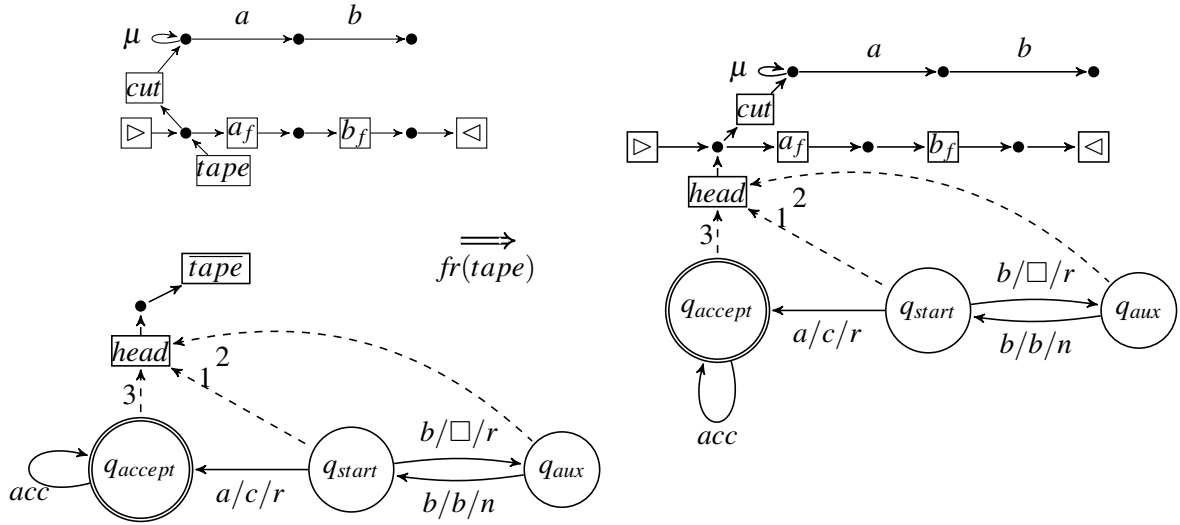
\begin{figure}[t]
  \begin{tikzpicture}
      \node at (0,2) (tape){ \begin{tikzpicture}[baseline=0.0cm,node distance=1.3cm,>=stealth',bend angle=45,auto]
  \node [inner sep = 1.5pt,circle,fill=black] at (1,0)(0) {};    
  \node [inner sep = 1.5pt,circle,fill=black] at (1,1.25)(1) {};

  \node [inner sep = 1.5pt,circle,fill=black] at (2.5,0)(2) {};
  \node [inner sep = 1.5pt,circle,fill=black] at (2.5,1.25)(3) {};

  \node [inner sep = 1.5pt,circle,fill=black] at (4,0)(4) {};
  \node [inner sep = 1.5pt,circle,fill=black] at (4,1.25)(5) {};

  \node [transition] at (1.5,-0.5)  (t1)  {$tape$}
  edge [post]                    (0);

  \node [transition] at (0.45,0.6)  (t2)  {$cut$}
  edge [pre]                    (0)
  edge [post]                   (1);
  \node [transition] at (0.34,0)  (t3)  {$\tmtapestart$}
  edge [post]                   (0);
  
  \path (1) edge [loop left] node{$\mu$}();

  \node [transition] at (1.75,0)  (t3)  {$\fsym{a}$}
  edge [pre]                    (0)
  edge [post]                   (2);
  \path (1) edge[->] node[above=.06]{$a$} (3);
  \node [transition] at (3.25,0)  (t3)  {$\fsym{b}$}
  edge [pre]                    (2)
  edge [post]                   (4);
  \path (3) edge[->] node[above=.06]{$b$} (5);

  \node [transition] at (4.8,0)  (t3)  {$\tmtapeend$}
  edge [pre]                   (4);

\end{tikzpicture}};
      \node at (0,-2) (tm){ \begin{tikzpicture}[->,>=stealth',shorten >=1pt,auto,node distance=2.8cm,
                    semithick]
  \tikzstyle{every state}=[draw=black,text=black]

  \node [inner sep = 1.5pt,circle,fill=black] at (1,0)(1) {};
  
  \node [transition] at (2,0.5)  (t2)  {$~\fuscomp{tape}~$}
  edge [pre]                    (1);
 
  \node[state,accepting]   (acc)  [below = 1.4 of 1]      {$q_{accept}$};

  \path (acc) edge[loop left]  node[below = 0.1] {$acc$} (acc);

  \node[state] [right of=acc] (start)                    {$q_{start}$};
  \node[state]         (q1) [right of=start] {$q_{aux}$};

  \node [transition] [below = 0.3 of 1]  (head)  {$head$}
  edge [pre,dashed] node [right = 0.01, very near start] {$1$}                    (start)
  edge [pre,bend left=28,dashed] node [below = 0.01, very near start] {$2$}                  (q1)
  edge [pre,dashed] node [below = 0.01, very near start] {$3~~~$}                  (acc)
  edge [post,dashed]                    (1);

  \path (start) edge[bend left=8]              node {$b/\tmblank/r$} (q1);
  \path (q1) edge[bend left=8]              node {$b/b/n$} (start);
  \path (start) edge              node {$a/c/r$} (acc);

\end{tikzpicture}};

      \node at (2,0) (arrowstart){};
      \node at (3,0) (arrowend){};

      \node at (8,0) (right){\begin{tikzpicture}[->,>=stealth',shorten >=1pt,auto,node distance=2.8cm,
                    semithick]
  \tikzstyle{every state}=[draw=black,text=black]

    \foreach \x in {1,3,5}{
     \node [inner sep = 1.5pt,circle,fill=black] at (\x,0)(\x) {};
  }
  \foreach \x in {2,4,6}{
     \node [inner sep = 1.5pt,circle,fill=black] at (\x,1)(\x) {};
  }
 
  \node [transition] at (1.5,0.5)  (t2)  {$cut$}
  edge [pre]                    (1)
  edge [post]                   (2);
  \path (2) edge [loop left] node{$\mu$}();

  \foreach \x/\z [evaluate=\x as \y using \x+2] in {2/a,4/b}{
    \path (\x) edge[->,shorten >= 5pt] node[above=.06]{$\z$} (\y);
  }
  \foreach \x/\z [evaluate=\x as \y using \x+2,evaluate=\x as \p using \x+1] in {1/\fsym{a},3/\fsym{b}}{    
    \node [transition] at (\p,0)  (t\x)  {$\z$}
      edge [pre]                    (\x)
      edge [post,shorten >= 5pt]                   (\y);
  }

  \node [transition] at (0,0)  (start)  {$\tmtapestart$}
  edge [post]                    (1);
  \node [transition] at (6,0)  (end)  {$\tmtapeend$}
  edge [pre]                    (5);
 
  \node[state,accepting]   (acc)  [below = 1.4 of 1]      {$q_{accept}$};

  \path (acc) edge[loop below]  node[below] {$acc$} (acc);

  \node[state] [right of=acc] (start)                    {$q_{start}$};
  \node[state]         (q1) [right of=start] {$q_{aux}$};

  \node [transition] [below = 0.3 of 1]  (head)  {$head$}
  edge [pre,dashed] node [right = 0.01, very near start] {$1$}                    (start)
  edge [pre,bend left=28,dashed] node [below = 0.01, very near start] {$2$}                  (q1)
  edge [pre,dashed] node [below = 0.01, very near start] {$3~~~$}                  (acc)
  edge [post,dashed]                    (1);

  \path (start) edge[bend left=8]              node {$b/\tmblank/r$} (q1);
  \path (q1) edge[bend left=8]              node {$b/b/n$} (start);
  \path (start) edge              node {$a/c/r$} (acc);

\end{tikzpicture}};
      
      \path (arrowstart) edge[-implies,thick,
        double distance=1.4pt,
        preaction = {decorate},
        postaction = {draw,line width=1.4pt, white,shorten >= 4.5pt}] node[below =1pt]{ $\fr(tape)$}  (arrowend);
  \end{tikzpicture}
  \caption{
    Fusion of the hypergraph representation of the Turing machine in Example~\ref{example:cdfg-TM-example} and the tape graph generated in Example~\ref{example:trans(PN)-3} yields the configuration $(q_{start}, \varepsilon, ab)$ wrt the input adjunct $ab$ and the permutation $q_{start}q_{aux}q_{accept}$.}
  \label{fig:cdfg-TM-example-configureation-tape-fusion}
\end{figure}

\begin{definition}\label{def:construction-hg_ci}
Let $c = (q, \alpha, \beta) \in \mathit{conf}(\TM)$, let $w\in \Omega^*$,
let $\sigma = q_1 \cdots q_{|Q|}$ be a permutation of $Q$ where $q_1 = q$.
Then $H = hg(\sigma, \alpha, \beta, w)$ is the \emph{hypergraph representation of the configuration} $c$ \emph{with adjunct} $w$ \emph{and permutation} $\sigma$, where $w$ is the terminal labeled string graph in the tape graph. $H$ is defined by $hg(\TM,\sigma) + \tg(\alpha, \beta, w)_{tape} \dder_{\fr(tape)} H$.
\end{definition}

\subsection{Simulating steps of a Turing machine by context-dependent fusion rules}

In order to simulate a step further connected components are needed.
These connected components encode the substitution of a symbol on the tape, the movement of the head and the transition to the next state.
In order to move the head to the left or to the right our construction takes both the current symbol and the symbol to the left of the head into account.
The relations of the Turing machine can be seen as replacing or deleting the current symbol and (maybe) inserting a new symbol left or right of the head.
In the graph representation this corresponds to deleting and inserting edges.
These deletions and insertions are done with respective fusions of complementary labeled hyperedges.
The hypergraphs in Figure~\ref{fig:cdfg-tm-1-l1}--\ref{fig:cdfg-tm-1-r1} are schematic drawings of the connected components used for simulating a step of a Turing machine.
The dots indicate that there are $|Q|$ vertices as sources.
Two complementary $head$ and $\fuscomp{head}$-hyperedges attached to the same vertices are part of this connected component, where the ordering of the source attachments of the $head$-hyperedge implements a permutation
such that the first and $i$th source are swapped.
Formally, these components are defined as follows.

\begin{figure}[t]
\centering
\begin{subfigure}[t]{0.24\textwidth}
\begin{tikzpicture}[baseline=0.3cm,node distance=1.3cm,>=stealth',bend angle=25,auto]

  \node [inner sep = 1.5pt,circle,fill=black] at (1,0)(s1) {};
  \node [inner sep = 1.5pt,circle,fill=black] at (2,0)(s2) {};
  \node [inner sep = 1.5pt,circle,fill=black] at (3,0)(s4) {};

  \node [inner sep = 1.5pt,circle,fill=black] at (2,1.75)(0) {};

  \node [inner sep = 1.5pt,circle,fill=black] at (2.75,2.5)(1) {};
  \node [inner sep = 1.5pt,circle,fill=black] at (3.5,1.75)(2) {};
  
  \node at (1.5,0) {$\dots$};    
  \node at (2.5,0) {$\dots$};    

  \node [transition] at (1,1)  (t1)  {$\fuscomp{head}$}
  edge [pre] node [left = 0.01, near start]{}          (s1)
  edge [pre] node [left = 0.01, near start]{$i$} (s2)
  edge [pre]                     (s4)
  edge [post]                    (0);

  \node [transition] at (3,1)  (t1)  {$head$}
  edge [pre]                     (s1)
  edge [pre] node [right = 0.01, near start]{$1$}                    (s2)
  edge [pre]                     (s4)
  edge [post]                    (0);

  \node [transition] at (1,1.75)  (x)  {$\fsym{\fuscomp{u}}$}
  edge [pre,bend right]                     (0)
  edge [post, bend left]                     (0);
  
  \node [transition] at (2.75,1.75) (zbar)  {$\fsym{\fuscomp{x}}$}
  edge [pre]                     (0)
  edge [post]                    (2);

  \node [transition] at (3.5,2.5)  (y)  {$\fsym{y}$}
  edge [post]                     (2)
  edge [pre]                    (1);

  \node [transition] at (2,2.5) (z)  {$\fsym{u}$}
  edge [pre]                     (0)
  edge [post]                    (1);
               
\end{tikzpicture}
\caption{$C(u, x/y/l, i)$}
\label{fig:cdfg-tm-1-l1}
\end{subfigure}
\begin{subfigure}[t]{0.26\textwidth}
\begin{tikzpicture}[baseline=0.3cm,node distance=1.3cm,>=stealth',bend angle=25,auto]

  \node [inner sep = 1.5pt,circle,fill=black] at (1,0)(s1) {};
  \node [inner sep = 1.5pt,circle,fill=black] at (2,0)(s2) {};
  \node [inner sep = 1.5pt,circle,fill=black] at (3,0)(s4) {};

  \node [inner sep = 1.5pt,circle,fill=black] at (2,1.75)(0) {};

  \node [inner sep = 1.5pt,circle,fill=black] at (3.5,1.75)(1) {};
  
  \node at (1.5,0) {$\dots$};    
  \node at (2.5,0) {$\dots$};    

  \node [transition] at (1,1)  (t1)  {$\fuscomp{head}$}
  edge [pre] node [left = 0.01, near start]{}          (s1)
  edge [pre] node [left = 0.01, near start]{$i$} (s2)
  edge [pre]                     (s4)
  edge [post]                    (0);

  \node [transition] at (3,1)  (t1)  {$head$}
  edge [pre]                     (s1)
  edge [pre] node [right = 0.01, near start]{$1$}                    (s2)
  edge [pre]                     (s4)
  edge [post]                    (0);

  \node [transition] at (2.75,1.75)  (t2)  {$\fsym{\fuscomp{x}}$}
  edge [pre]                    (0)
  edge [post]                   (1);
  \node [transition] at (2.75,2.25)  (t2)  {$\fsym{y}$}
  edge [pre]                    (0)
  edge [post]                   (1);
               
\end{tikzpicture}
\caption{$C(u, x/y/n, i)$}
\label{fig:cdfg-tm-1-n1}
\end{subfigure}
\begin{subfigure}[t]{0.28\textwidth}
\begin{tikzpicture}[baseline=0.3cm,node distance=1.3cm,>=stealth',bend angle=25,auto]

  \node [inner sep = 1.5pt,circle,fill=black] at (1,0)(s1) {};
  \node [inner sep = 1.5pt,circle,fill=black] at (2,0)(s2) {};
  \node [inner sep = 1.5pt,circle,fill=black] at (3,0)(s4) {};

  \node [inner sep = 1.5pt,circle,fill=black] at (2,1.75)(0) {};

  \node [inner sep = 1.5pt,circle,fill=black] at (1.25,2.5)(1) {};
  \node [inner sep = 1.5pt,circle,fill=black] at (0.5,1.75)(2) {};
  
  \node at (1.5,0) {$\dots$};    
  \node at (2.5,0) {$\dots$};    

  \node [transition] at (1,1)  (t1)  {$\fuscomp{head}$}
  edge [pre] node [left = 0.01, near start]{}          (s1)
  edge [pre] node [left = 0.01, near start]{$i$} (s2)
  edge [pre]                     (s4)
  edge [post]                    (0);

  \node [transition] at (3,1)  (t1)  {$head$}
  edge [pre]                     (s1)
  edge [pre] node [right = 0.01, near start]{$1$}                    (s2)
  edge [pre]                     (s4)
  edge [post]                    (0);

  \node [transition] at (3,1.75)  (x)  {$\fsym{\fuscomp{x}}$}
  edge [pre,bend right]                     (0)
  edge [post, bend left]                     (0);

  \node [transition] at (1.25,1.75) (zbar)  {$\fsym{\fuscomp{u}}$}
  edge [post]                     (0)
  edge [pre]                    (2);

  \node [transition] at (2,2.5)  (y)  {$\fsym{y}$}
  edge [post]                     (0)
  edge [pre]                    (1);

  \node [transition] at (0.5,2.5) (z)  {$\fsym{u}$}
  edge [pre]                     (2)
  edge [post]                    (1);
               
\end{tikzpicture}
\caption{$C(u, x/y/r, i)$}
\label{fig:cdfg-tm-1-r1}
\end{subfigure}
\caption{Schematic drawings of connected components used for simulating a step of a Turing machine}
\label{Fig:cdfg-TM-components-cl-cn-cr-acc}
\end{figure}
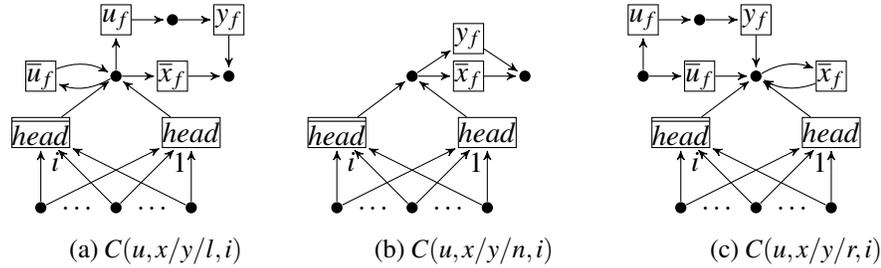

\begin{definition}\label{def:construction-CDFG(pts)-step-simulation}
Let $\{head\} + \fsym{\Gamma}$ be a fusion alphabet with $type(head) = (|Q|,1)$ and $type(x) = (1,1)$ for each $x \in \fsym{\Gamma}$.
Let $\Lambda = \{ x/y/dir \mid x,y \in \Gamma, dir \in \{l,n,r\}\}$.

Define
for each
$u \in \fsym{\Gamma}, x/y/l \in \Lambda, i \in \{1,\ldots, |Q|\}$ the connected component  
$C(u,x/y/l,i) =\linebreak ( \{v_1, \ldots, v_{|Q|+3} \}, \{ head, \fuscomp{head}, \fuscomp{u}, u, \fuscomp{x}, y \},\sE,\tE,\lE)$
where
$\sE(head) = v_1 v_2\cdots v_i \cdots  v_{|Q|}$, $\tE(head) = v_{|Q|+1}$,\linebreak $\lE(head) = head$,
$\sE(\fuscomp{head}) = v_i v_2\cdots v_1 \cdots v_{|Q|}$, $\tE(\fuscomp{head}) = v_{|Q|+1}$, $\lE(\fuscomp{head}) = \fuscomp{head}$,
$\sE(\fuscomp{u}) = \tE(\fuscomp{u}) = v_{|Q|+1}$, $\lE(\fuscomp{u}) = \fsym{\fuscomp{u}}$,
$\sE(u) = v_{|Q|+1}$, $\tE(u) = v_{|Q|+3}$, $\lE(u) = \fsym{u}$,
$\sE(\fuscomp{x}) = v_{|Q|+1}$, $\tE(\fuscomp{x}) = v_{|Q|+2}$, $\lE(\fuscomp{x}) = \fsym{\fuscomp{x}}$,
$\sE(y) = v_{|Q|+2}$, $\tE(y) = v_{|Q|+3}$, $\lE(y) = \fsym{y}$
.

Define for each
$x/y/n \in \Lambda, i \in \{1,\ldots, |Q|\}$ the connected component 
$C(u,x/y/n,i) =\linebreak ( \{v_1, \ldots, v_{|Q|+2} \}, \{ head, \fuscomp{head}, \fuscomp{x}, y \},\sE,\tE,\lE)$
where
$\sE(head) = v_1 v_2\cdots v_i \cdots  v_{|Q|}$, $\tE(head) = v_{|Q|+1}$,\linebreak $\lE(head) = head$,
$\sE(\fuscomp{head}) = v_i v_2\cdots v_1 \cdots v_{|Q|}$, $\tE(\fuscomp{head}) = v_{|Q|+1}$, $\lE(\fuscomp{head}) = \fuscomp{head}$,
$\sE(\fuscomp{x}) = \sE(y) v_{|Q|+1}$, $\tE(\fuscomp{x}) = \tE(y) = v_{|Q|+2}$, $\lE(\fuscomp{x}) = \fsym{\fuscomp{x}}$, $\lE(y) = \fsym{y}$.

Define for each
$u \in \fsym{\Gamma}, x/y/r \in \Lambda, i \in \{1,\ldots, |Q|\}$ the connected component
$C(u,x/y/r,i) =\linebreak ( \{v_1, \ldots, v_{|Q|+3} \}, \{ head, \fuscomp{head}, \fuscomp{x}, u, \fuscomp{u}, y \},\sE,\tE,\lE)$
where
$\sE(head) = v_1 v_2\cdots v_i \cdots  v_{|Q|}$, $\tE(head) = v_{|Q|+2}$,\linebreak $\lE(head) = head$,
$\sE(\fuscomp{head}) = v_i v_2\cdots v_1 \cdots v_{|Q|}$, $\tE(\fuscomp{head}) = v_{|Q|+2}$, $\lE(\fuscomp{head}) = \fuscomp{head}$,
$\sE(\fuscomp{x}) = \tE(\fuscomp{x}) = v_{|Q|+2}$, $\lE(\fuscomp{x}) = \fsym{\fuscomp{x}}$,
$\sE(u) = v_{|Q|+1}$, $\tE(u) = v_{|Q|+3}$, $\lE(u) = \fsym{u}$,
$\sE(\fuscomp{u}) = v_{|Q|+1}$, $\tE(\fuscomp{u}) = v_{|Q|+2}$, $\lE(\fuscomp{u}) = \fsym{\fuscomp{u}}$,
$\sE(y) = v_{|Q|+2}$, $\tE(e_4) = v_{|Q|+3}$, $\lE(y) = \fsym{y}$
.
\end{definition}

In order to simulate a step of a Turing machine context-dependent fusion rules are needed.
Some of the context conditions derive directly from the semantics of a Turing machine.
For example, the step $(p, \alpha u, x, \beta) \tmstep{\TM} (q, \alpha, u, y \beta)$
can only be applied if $(p,x,y,l,q) \in \Delta$, the Turing machine is in state $p$ and reads the symbol $x$.
Other context conditions are needed because (context-dependent) fusion rules can only consume two complementary hyperedges in one derivation step but a step of a Turing machine is much more complicated (head movement, tape manipulation, state transition). Hence, several rules and rule applications are needed to simulate such a step.
Furthermore, in our construction positive and negative context conditions are needed to restrict the application to obtain a correct and sound simulation.
The set of context-dependent fusion rules $P_{\Delta}$ specified in Definition~\ref{def:construction-CDFG(pts)-step-simulation-rules} contains rules with respect to the fusion symbol $head$ and rules with respect to the fusion symbol $x$ for each $x \in \Gamma$.
The first are used to fuse the $head$-hyperedge in the graph representation of a configuration with the correct connected component used for simulating the step of the Turing machine and perform the state transition.
The latter are used to modify the tape and move the head correctly.

\begin{definition}\label{def:construction-CDFG(pts)-step-simulation-rules}
Define $P_{\Delta}$ as the following set of context-dependent fusion rules.
\[
P_{\Delta} = \{ \Delta(u,\lambda) \mid  u \in \Gamma, \lambda \in \Lambda \}
\cup \{fuse\_2in(x), fuse\_2out(x), fuse\_loop\_in(x), fuse\_loop\_out(x) \mid x \in \Gamma\}
\]

$\Delta(u,\lambda) = (\fr(head), \{fr(head) \to  PC(u,\lambda,j) + C(u,\lambda,j) \},$ $\{fr(head) \to  twoin(u) + \fuscomp{h}^\bullet \mid u \in \Gamma\}) \cup \{fr(head) \to  twoout(u) + \fuscomp{h}^\bullet \mid u \in \Gamma\})$,
where the morphism in the positive context maps the $head$-hyperedge to the one in $PC(u,\lambda,j)$, depicted in Figure~\ref{fig:cdfg-tm-2a}, and the $\fuscomp{head}$-hyperedge to the one in $C(u,\lambda,j)$.
Note that, the connected component $C(u,\lambda,j)$ is induced by the $j$th source of the $head$-hyperedge and the parameters $u,\lambda$.
$twoin(u)$, $twoout(u)$ and $\fuscomp{h}^\bullet$ are depicted in Figure~\ref{fig:cdfg-tm-2b},~\ref{fig:cdfg-tm-2c} and \ref{fig:cdfg-tm-2headbullet}, respectively.

\begin{figure}[t]
\begin{subfigure}[t]{0.25\textwidth}
\begin{tikzpicture}[baseline=0.3cm,node distance=1.3cm,>=stealth',bend angle=45,auto]

  \node [inner sep = 1.5pt,circle,fill=black] at (0,0)(s2) {};
  \node [inner sep = 1.5pt,circle,fill=black] at (1.85,0)(s3) {};
  \node [inner sep = 1.5pt,circle,fill=black] at (2.6,0)(s4) {};    
  \node [inner sep = 1.5pt,circle,fill=black] at (1.25,1.7)(top0) {};    
  \node [inner sep = 1.5pt,circle,fill=black] at (2.6,1.7)(top1) {};    
  \node [inner sep = 1.5pt,circle,fill=black] at (-0.25,1.7)(top2) {};    
  
  \node at (2.2,0) {$\dots$};    

  \node [transition] at (1.25,1)  (t1)  {$head$}
  edge [pre] node [left = 0.1, very near start]{$1$}        (s2)
  edge [pre] node [left = 0.1,  near start]{$j$}        (s3)
  edge [pre] (s4)
  edge [post]                    (top0);

  \node [transition] at (1.9,1.7)  (a)  {$\fsym{x}$}
  edge [pre]        (top0)
  edge [post]       (top1);

  \node [transition] at (0.5,1.7)  (a)  {$\fsym{u}$}
  edge [pre]        (top2)
  edge [post]       (top0);

  \path (s2) edge[->] node[below=.001]{~~$x/y/dir$} (s3);    
\end{tikzpicture}
\caption{$PC(u,x/y/dir,j)$}
\label{fig:cdfg-tm-2a}
\end{subfigure}
\begin{subfigure}[t]{0.16\textwidth}
\begin{tikzpicture}[baseline=0.3cm,node distance=1.3cm,>=stealth',bend angle=45,auto]

  \node [inner sep = 1.5pt,circle,fill=black] at (0.85,0)(s2) {};
  \node [inner sep = 1.5pt,circle,fill=black] at (1.65,0)(s3) {};

  \node [inner sep = 1.5pt,circle,fill=black] at (1.25,1.7)(top0) {};    

  \node [inner sep = 1.5pt,circle,fill=black] at (-0.5,2)(top1) {};    
  \node [inner sep = 1.5pt,circle,fill=black] at (-0.5,1.4)(top2) {};    
  
  \node at (1.25,0) {$\dots$};    

  \node [transition] at (1.25,1)  (t1)  {$head$}
  edge [pre]         (s2)
  edge [pre]         (s3)
  edge [post]                    (top0);

  \node [transition] at (0.25,1.4)  (a)  {$\fsym{u}$}
  edge [pre]        (top2)
  edge [post]       (top0);

    \node [transition] at (0.25,2)  (a)  {$\fsym{\fuscomp{u}}$}
  edge [pre]        (top1)
  edge [post]       (top0);

\end{tikzpicture}
\caption{$twoin(u)$}
\label{fig:cdfg-tm-2b}
\end{subfigure}
\begin{subfigure}[t]{0.16\textwidth}
\begin{tikzpicture}[baseline=0.3cm,node distance=1.3cm,>=stealth',bend angle=45,auto]

  \node [inner sep = 1.5pt,circle,fill=black] at (0.85,0)(s2) {};
  \node [inner sep = 1.5pt,circle,fill=black] at (1.65,0)(s3) {};

  \node [inner sep = 1.5pt,circle,fill=black] at (1.25,1.7)(top0) {};    

  \node [inner sep = 1.5pt,circle,fill=black] at (3,2)(top1) {};    
  \node [inner sep = 1.5pt,circle,fill=black] at (3,1.4)(top2) {};    
  
  \node at (1.25,0) {$\dots$};    

  \node [transition] at (1.25,1)  (t1)  {$head$}
  edge [pre]         (s2)
  edge [pre]         (s3)
  edge [post]                    (top0);

  \node [transition] at (2.25,1.4)  (a)  {$\fsym{u}$}
  edge [pre]        (top0)
  edge [post]       (top2);

    \node [transition] at (2.25,2)  (a)  {$\fsym{\fuscomp{u}}$}
  edge [pre]        (top0)
  edge [post]       (top1);

\end{tikzpicture}
\caption{$twoout(u)$}
\label{fig:cdfg-tm-2c}
\end{subfigure}
\begin{subfigure}[t]{0.09\textwidth}
\begin{tikzpicture}[baseline=0.3cm,node distance=1.3cm,>=stealth',bend angle=45,auto]

  \node [inner sep = 1.5pt,circle,fill=black] at (0,0)(s2) {};
  \node [inner sep = 1.5pt,circle,fill=black] at (0.7,0)(s3) {};

  \node [inner sep = 1.5pt,circle,fill=black] at (0.35,1.7)(top0) {};    
  
  \node at (.35,0) {$\dots$};    

  \node [transition] at (.35,1)  (t1)  {$\fuscomp{head}$}
  edge [pre]         (s2)
  edge [pre]         (s3)
  edge [post]                    (top0);

\end{tikzpicture}
\caption{$\fuscomp{h}^\bullet$}
\label{fig:cdfg-tm-2headbullet}
\end{subfigure}
\begin{tabular}{c}
\begin{subfigure}{0.2\textwidth}
\begin{tikzpicture}[baseline=0.3cm,node distance=1.3cm,>=stealth',bend angle=25,auto]

  \node [inner sep = 1.5pt,circle,fill=black] at (2,0)(0) {};
  \node [inner sep = 1.5pt,circle,fill=black] at (0,0)(1) {};

  \node [transition] at (3,0)  (x)  {$\fsym{\fuscomp{x}}$}
  edge [pre,bend right]                     (0)
  edge [post, bend left]                     (0);

  \node [transition] at (1,0) (zbar)  {$\fsym{x}$}
  edge [post]                     (0)
  edge [pre]                    (1);              
\end{tikzpicture}
\caption{$loop/in(x)$}
\label{fig:cdfg-tm-2d}
\end{subfigure}\\
\begin{subfigure}{0.2\textwidth}
\begin{tikzpicture}[baseline=0.3cm,node distance=1.3cm,>=stealth',bend angle=25,auto]

  \node [inner sep = 1.5pt,circle,fill=black] at (2,0)(0) {};
  \node [inner sep = 1.5pt,circle,fill=black] at (4,0)(1) {};

  \node [transition] at (1,0)  (x)  {$\fsym{\fuscomp{x}}$}
  edge [pre,bend right]                     (0)
  edge [post, bend left]                     (0);

  \node [transition] at (3,0) (zbar)  {$\fsym{x}$}
  edge [post]                     (1)
  edge [pre]                    (0);              
\end{tikzpicture}
\caption{$loop/out(x)$}
\label{fig:cdfg-tm-2e}
\end{subfigure}
\end{tabular}

\begin{subfigure}{0.22\textwidth}
\begin{tikzpicture}[baseline=0.3cm,node distance=1.3cm,>=stealth',bend angle=45,auto]

  \node [inner sep = 1.5pt,circle,fill=black] at (0.85,0)(s2) {};
  \node [inner sep = 1.5pt,circle,fill=black] at (1.65,0)(s3) {};

  \node [inner sep = 1.5pt,circle,fill=black] at (1.25,1.7)(top0) {};    

  \node [inner sep = 1.5pt,circle,fill=black] at (-0.5,1.9)(top1) {};    
  \node [inner sep = 1.5pt,circle,fill=black] at (-0.5,1.5)(top2) {};    
  
  \node at (1.25,0) {$\dots$};    

  \node [transition] at (1.75,1)  (t1)  {$\fuscomp{head}$}
  edge [pre]         (s2)
  edge [pre]         (s3)
  edge [post]                    (top0);

  \node [transition] at (0.75,1)  (t1)  {$head$}
  edge [pre]         (s2)
  edge [pre]         (s3)
  edge [post]                    (top0);

  \node [transition] at (0.25,1.5)  (a)  {$\fsym{x}$}
  edge [pre]        (top2)
  edge [post]       (top0);

  \node [transition] at (0.25,1.9)  (a)  {$\fsym{\fuscomp{x}}$}
  edge [pre]        (top1)
  edge [post]       (top0);

\end{tikzpicture}
\caption{$twoin2h(x)$}
\label{fig:cdfg-tm-2bb}
\end{subfigure}
\begin{subfigure}{0.22\textwidth}
\begin{tikzpicture}[baseline=0.3cm,node distance=1.3cm,>=stealth',bend angle=45,auto]

  \node [inner sep = 1.5pt,circle,fill=black] at (0.85,0)(s2) {};
  \node [inner sep = 1.5pt,circle,fill=black] at (1.65,0)(s3) {};

  \node [inner sep = 1.5pt,circle,fill=black] at (1.25,1.7)(top0) {};    

  \node [inner sep = 1.5pt,circle,fill=black] at (2.8,1.9)(top1) {};    
  \node [inner sep = 1.5pt,circle,fill=black] at (2.8,1.5)(top2) {};    
  
  \node at (1.25,0) {$\dots$};    

  \node [transition] at (1.75,1)  (t1)  {$\fuscomp{head}$}
  edge [pre]         (s2)
  edge [pre]         (s3)
  edge [post]                    (top0);

  \node [transition] at (0.75,1)  (t1)  {$head$}
  edge [pre]         (s2)
  edge [pre]         (s3)
  edge [post]                    (top0);

  \node [transition] at (2.0,1.5)  (a)  {$\fsym{x}$}
  edge [pre]        (top0)
  edge [post]       (top2);

    \node [transition] at (2.0,1.9)  (a)  {$\fsym{\fuscomp{x}}$}
  edge [pre]        (top0)
  edge [post]       (top1);

\end{tikzpicture}
\caption{$twoin2h(x)$}
\label{fig:cdfg-tm-2cb}
\end{subfigure}
\begin{tabular}{c}
\begin{subfigure}{0.33\textwidth}
\begin{tikzpicture}[baseline=0.3cm,node distance=1.3cm,>=stealth',bend angle=25,auto]

  \node [inner sep = 1.5pt,circle,fill=black] at (2,0)(0) {};
  \node [inner sep = 1.5pt,circle,fill=black] at (0,0)(1) {};
  \node [inner sep = 1.5pt,circle,fill=black] at (0,0.5)(2) {};

  \node [transition] at (3,0)  (zloop)  {$\fsym{\fuscomp{z}}$}
  edge [pre,bend right]                     (0)
  edge [post, bend left]                     (0);

  \node [transition] at (1,0.5) (x)  {$\fsym{\fuscomp{x}}$}
  edge [post]                     (0)
  edge [pre]                    (2);              

  \node [transition] at (1,0) (x)  {$\fsym{x}$}
  edge [post]                     (0)
  edge [pre]                    (1);              
\end{tikzpicture}
\caption{$twoinextraloop(x,z)$}
\label{fig:cdfg-tm-2g}
\end{subfigure}\\
\begin{subfigure}{0.33\textwidth}
\begin{tikzpicture}[baseline=0.3cm,node distance=1.3cm,>=stealth',bend angle=25,auto]

  \node [inner sep = 1.5pt,circle,fill=black] at (2,0)(0) {};
  \node [inner sep = 1.5pt,circle,fill=black] at (4,0)(1) {};
  \node [inner sep = 1.5pt,circle,fill=black] at (4,0.5)(2) {};

  \node [transition] at (1,0)  (zloop)  {$\fsym{\fuscomp{z}}$}
  edge [pre,bend right]                     (0)
  edge [post, bend left]                     (0);

  \node [transition] at (3,0.5) (xbar)  {$\fsym{\fuscomp{x}}$}
  edge [post]                     (2)
  edge [pre]                    (0);              

  \node [transition] at (3,0) (x)  {$\fsym{x}$}
  edge [post]                     (1)
  edge [pre]                    (0);              
\end{tikzpicture}
\caption{$twooutextraloop(x,z)$}
\label{fig:cdfg-tm-2h}
\end{subfigure}
\end{tabular}
\begin{subfigure}[t]{0.23\textwidth}
\begin{tikzpicture}[baseline=0.3cm,node distance=1.3cm,>=stealth',bend angle=25,auto]

  \node[circle,draw=black,radius=1mm] at (2,0)(0) {$v_1$};

  \node [circle,draw=black] at (3,0.85)(1) {$v_2$};
  \node [circle,draw=black] at (4,0)(2) {$v_3$};
    
  \node [transition] at (3,0) (zbar)  {$\fsym{\fuscomp{x}}$}
  edge [pre]                     (0)
  edge [post]                    (2);

  \node [transition] at (4,0.85)  (y)  {$\fsym{x}$}
  edge [post]                     (2)
  edge [pre]                    (1);

  \node [transition] at (2,0.85) (z)  {$\fsym{x}$}
  edge [pre]                     (0)
  edge [post]                    (1);
               
\end{tikzpicture}
\caption{$tri(x)$}
\label{fig:cdfg-tm-2f}
\end{subfigure}
\begin{subfigure}[t]{0.36\textwidth}
\begin{tikzpicture}[baseline=0.3cm,node distance=1.3cm,>=stealth',bend angle=25,auto]

  \node[inner sep = 1.5pt,circle,fill=black] at (0,0)(3) {};
  \node[circle,draw=black] at (2,0)(0) {$v$};

  \node [inner sep = 1.5pt,circle,fill=black] at (3,0.75)(1) {};
  \node [inner sep = 1.5pt,circle,fill=black] at (4,0.2)(2) {};

  \node [inner sep = 1.5pt,circle,fill=black] at (4,-0.25)(4) {};

  \node [transition] at (1,0) (zbar)  {$\fsym{x}$}
  edge [pre]                     (3)
  edge [post]                    (0);
  
  \node [transition] at (1,0.75)  (x)  {$\fsym{\fuscomp{x}}$}
  edge [pre,bend right]                     (0)
  edge [post, bend left]                     (0);
 
  \node [transition] at (3,0.2) (zbar)  {$\fsym{\fuscomp{z}}$}
  edge [pre]                     (0)
  edge [post]                    (2);

  \node [transition] at (3,-0.25) (zbar)  {$\fsym{z}$}
  edge [pre]                     (0)
  edge [post]                    (4);

  \node [transition] at (4,0.75)  (y)  {$\fsym{y}$}
  edge [post]                     (2)
  edge [pre]                    (1);

  \node [transition] at (2,0.75) (z)  {$\fsym{x}$}
  edge [pre]                     (0)
  edge [post]                    (1);
               
\end{tikzpicture}
\caption{$PCloopin(x)$}
\label{fig:cdfg-tm-2i}
\end{subfigure}
\begin{subfigure}[t]{0.36\textwidth}
\begin{tikzpicture}[baseline=0.3cm,node distance=1.3cm,>=stealth',bend angle=25,auto]

  \node[circle,draw=black] at (2,0)(0) {$v$};

  \node [inner sep = 1.5pt,circle,fill=black] at (1,0.75)(1) {};
  \node [inner sep = 1.5pt,circle,fill=black] at (0,0.2)(2) {};
  \node [inner sep = 1.5pt,circle,fill=black] at (0,-0.25)(3) {};
  \node [inner sep = 1.5pt,circle,fill=black] at (4,0)(4) {};
  
  \node [transition] at (3,0.75)  (x)  {$\fsym{\fuscomp{x}}$}
  edge [pre,bend right]                     (0)
  edge [post, bend left]                     (0);
  
  \node [transition] at (1,0.2) (zbar)  {$\fsym{\fuscomp{z}}$}
  edge [post]                     (0)
  edge [pre]                    (2);

  \node [transition] at (1,-0.25) (zbar)  {$\fsym{z}$}
  edge [post]                     (0)
  edge [pre]                    (3);

  \node [transition] at (2,0.75)  (y)  {$\fsym{y}$}
  edge [post]                     (0)
  edge [pre]                    (1);

  \node [transition] at (0,0.75) (z)  {$\fsym{z}$}
  edge [pre]                     (2)
  edge [post]                    (1);

  \node [transition] at (3,0) (z)  {$\fsym{x}$}
  edge [pre]                     (0)
  edge [post]                    (4);
\end{tikzpicture}
\caption{$PCloopout(x)$}
\label{fig:cdfg-tm-2j}
\end{subfigure}
\caption{Some contexts for applying rules simulating a step}
\label{Fig:cdfg-TM-contexts}
\end{figure}
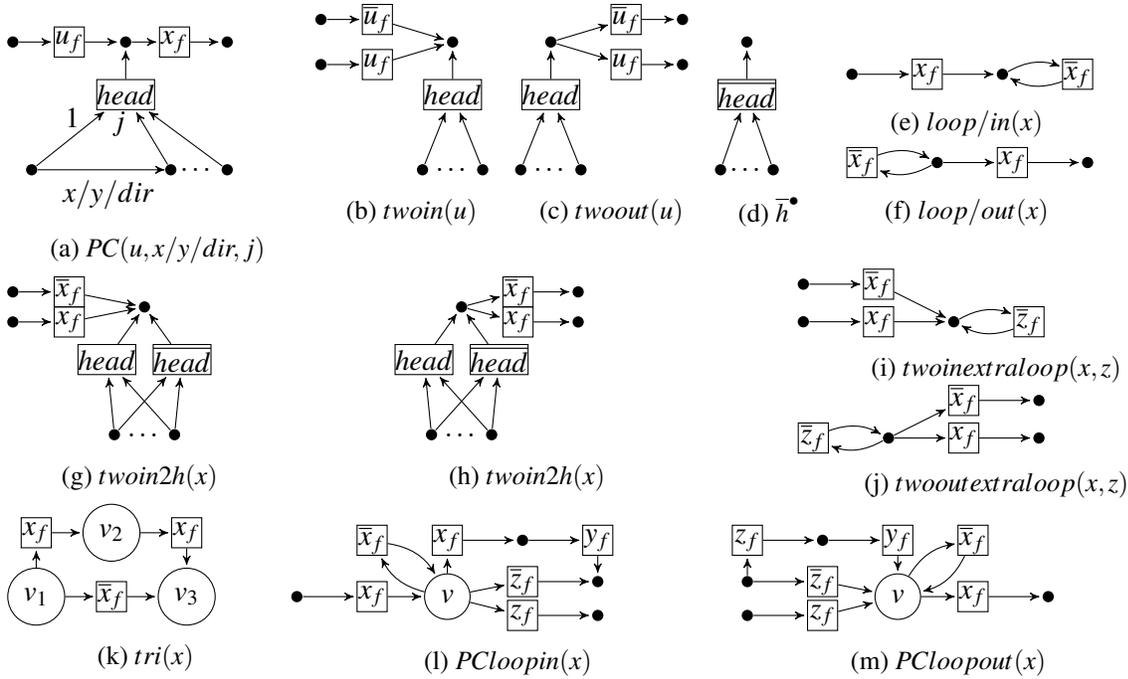


$fuse\_2in(x) = (\fr(x), \{\fr(x) \to twoin(x) \}, \{\fr(x) \to twoin2h(x), \fr(x) \to loop/in(x),\fr(x) \to tri(x) \} \cup \{\fr(x) \to twoinextraloop(x,z) \mid z\in \Gamma \})$,
where $\fr(x) \to tri(x)$ maps the $x$-hyperedge to $e \in E_{tri(x)}$ with $s(e) = v_2$ and $t(e) = v_3$; $\fr(x) \to twoinextraloop(x,z)$ maps the $\fuscomp{x}$-hyperedge to the one above the the $x$-hyperedge (which is relevant for the case $z=x$).
The respective hypergraphs used in the context conditions are depicted in Figure~\ref{Fig:cdfg-TM-contexts}.

$fuse\_2out(x) = (\fr(x), \{\fr(x) \to twoout(x) \}, \{\fr(x) \to twoout2h(x), \fr(x) \to loop/out(x),\fr(x) \to tri(x) \} \cup \{\fr(x) \to twooutextraloop(x,z) \mid z\in \Gamma \})$
is analog but $\fr(x) \to tri(x)$ maps the $x$-hyperedge to $e \in E_{tri(x)}$ with $s(e) = v_1$ and $t(e) = v_2$.

$fuse\_loop\_in(x) = (\fr(x), \{ \fr(x) \to PCloopin(x) \}, \emptyset)$,
where $PCloopin(x)$ is depicted in Figure~\ref{fig:cdfg-tm-2i} and
in the positive context condition the $x$-hyperedge matches the one with target $v$ and the $\fuscomp{x}$-hyperedge matches the one with source and target $v$.

$fuse\_loop\_out(x) = (\fr(x), \{ \fr(x) \to PCloopout(x) \}, \emptyset)$,
is analog but the $x$-hyperedge matches the one with source $v$.
$PCloopout(x)$ is depicted in Figure~\ref{fig:cdfg-tm-2j}.

\end{definition}

\begin{remark}
The rules in the latter subset of $P_\Delta$ are constructed in such a way that the two complementary hyperedges must be attached to the same vertex and that two rules are never applicable at the same time with respect to the same connected component.
This is achieved by defining $fuse\_2in(x)$ and $fuse\_loop\_in(x)$ as well as 
$fuse\_2out(x)$ and $fuse\_loop\_out(x)$ in such a way that the positive context condition of the first is reflected in the negative context condition of the other.
Further, in order to distinguish the applicability of $fuse\_loop\_in(x)$ and $fuse\_loop\_out(x)$ if in-going and out-going edges (wrt some vertex where the loop is attached) are labeled with the same symbol,
the number of out-going (in-going) edges, respectively, is of relevance.
If the number of out-going (in-going) edges is 4, then the loop must be fused with the in-going (out-going) edge, respectively.
$fuse\_loop\_in(x)$ and $fuse\_loop\_out(x)$ do not need negative context conditions
because the positive context does not occur in the $C$-components of Definition~\ref{def:construction-CDFG(pts)-step-simulation}.
\end{remark}

The following example illustrates the simulation of the transition step for $(q_{start}, a,c,r, q_{accept})$.
It works analogously for other transition steps.

\begin{example}
  Consider the two connected components $H = hg(q_{start} q_{aux} q_{accept}, d, ab, ab)$ and $C = C(d,\allowbreak a/c/r, 3)$. Then $H' = hg(q_{accept} q_{aux} q_{start}, dc, b, ab)$ can be derived as illustrated in Figure~\ref{fig:cdfg-tm-example-step-simulation}.
  \begin{figure}[t]
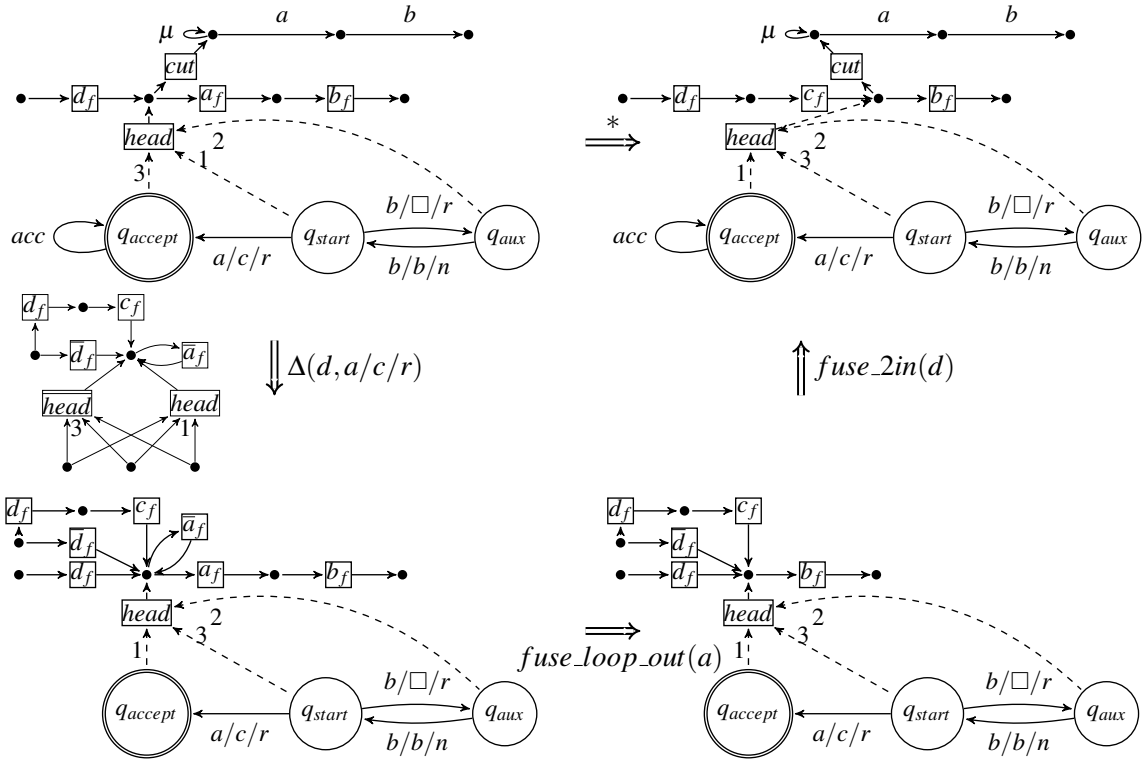

    \centering
     \begin{tikzpicture}
 \node at (-4,0) (init){\scalebox{0.85}{\input{tikz/cdfg-tm-example-simple-with-tape-ext-reduced}}};
 \node at (4,0) (result){\scalebox{0.85}{\input{tikz/cdfg-tm-example-simple-with-tape-ext-2-nc}}};

 \node at (-6,-3.2) (auxcomponent) {\scalebox{0.85}{\input{tikz/cdfg-tm-1-r1-example-2}}};

 \node at (-4,-6.5) (intermediate1){\scalebox{0.85}{\input{tikz/cdfg-tm-example-simple-with-tape-ext-3-nocolor}}};
 \node at (4,-6.5) (intermediate2){\scalebox{0.85}{\input{tikz/cdfg-tm-example-simple-with-tape-ext-4-nocolor}}};

    \node at (0,0) (arrowstart1){};
    \node at (1,0) (arrowend1){};
     
      \node at (-4,-2.5) (arrowstart2){};
      \node at (-4,-3.5) (arrowend2){};

      \node at (3,-3.5) (arrowstart3){};
      \node at (3,-2.5) (arrowend3){};

      \node at (0,-6.5) (arrowstart4){};
      \node at (1,-6.5) (arrowend4){};

      \path (arrowstart1) edge[-implies,thick,
        double distance=1.4pt,
        preaction = {decorate},
        postaction = {draw,line width=1.4pt, white,shorten >= 4.5pt}] node[above =1pt]{ $*$}  (arrowend1);
      
      \path (arrowstart2) edge[-implies,thick,
      double distance=1.4pt,
      preaction = {decorate},
      postaction = {draw,line width=1.4pt, white,shorten >= 4.5pt}] node[right =1pt]{ $\Delta(d, a/c/r)$} (arrowend2);

      \path (arrowstart3) edge[-implies,thick,
        double distance=1.4pt,
        preaction = {decorate},
        postaction = {draw,line width=1.4pt, white,shorten >= 4.5pt}] node[right =1pt]{ $fuse\_2in(d)$}  (arrowend3);

\path (arrowstart4) edge[-implies,thick,
   double distance=1.4pt,
   preaction = {decorate},
   postaction = {draw,line width=1.4pt, white,shorten >= 4.5pt}] node[below =1pt]{ $~~fuse\_loop\_out(a)$} (arrowend4);

 \end{tikzpicture}
    \caption{Simulating the transition step for $(q_{start}, a,c,r, q_{accept})$}
\label{fig:cdfg-tm-example-step-simulation}
\end{figure}  
  The context-dependent fusion rule $\Delta(d,a/c/r)$ can be applied by matching the $head$-hyperedge in $H$ and the $\fuscomp{head}$-hyperedge in $C$.
The two complementary hyperedges are fused and due to the $head$-hyperedge in $C$ a new $head$-hyperedge reconstructed; its first source is $q_{accept}$.
Furthermore, the application attaches additional vertices and hyperedges to the tape graph.
The resulting connected component is depicted in the lower left in Figure~\ref{fig:cdfg-tm-example-step-simulation} (the $acc$-loop and the adjunct, consisting of the $cut$-hyperedge and terminal and marker string graph, is omitted in order to clarify the drawing).
The $\fsym{\fuscomp{a}}$-hyperedge is then fused with the out-going $\fsym{a}$-hyperedge by an application of $fuse\_loop\_out(a)$ with the result that the source and the target vertex of the $\fsym{a}$-hyperedge are glued together yielding the connected component depicted in the lower right in Figure~\ref{fig:cdfg-tm-example-step-simulation} (again omitting the $acc$-loop and the adjunct).
Afterwards, $fuse\_2in(d)$ can be applied to the $\fsym{\fuscomp{d}}$- and $\fsym{d}$-hyperedge.
The resulting connected component is $H'$.
\end{example}

We will later prove this one-to-one correspondence of a transition step in $\TM$ and a particular derivation sequence concerning hypergraph representations of respective configurations wrt the same adjunct in the corresponding context-dependent fusion grammar $\CDFG(\TM)$.

\subsection{Construction of a context-dependent fusion grammar corresponding to a Turing machine}

Now we combine the previous constructions in a context-dependent fusion grammar.
In our construction two additional connected components $Acc$ and $\encsg{(\fuscomp{\tmblank})}$ and four additional (context-dependent) fusion rules are needed (1)~for connecting the initial hypergraph representation of the Turing machine with the generated tape graph (as described in Definition~\ref{def:construction-hg_ci}), (2)~for manipulating a hypergraph representation of an accepting configuration such that it indicates that an input string is accepted, (3)~for disconnecting\footnote{Disconnecting $\sg(w)_{\mu}$ uses similar ideas as discussed in Proposition~\ref{prop:cdfg-tape-graph-3} but uses context conditions.} the terminal and marker labeled string graph $\sg(w)_{\mu}$ for some $w \in \Omega^*$, and (4)~for removing a $\tmblank$-symbol from the left on the tape (for technical reasons). 

\begin{definition}\label{def:construction-CDFG(pts)}
  Let $TM = (Q,\Omega, \Gamma, \Delta)$ be a Turing machine.
Let $\CDFG_{\tg}(\Omega, \Gamma) = (Z_{\tg}, F_{\tg}, \{\mu\}, \Omega, P_{\tg})$ be the corresponding context-dependent fusion grammar generating tape graphs (cf. Definition~\ref{def:construction-CDFG-tg}),
let $C(u,\lambda,i)$ be the connected component defined in Definition~\ref{def:construction-CDFG(pts)-step-simulation} and $P_{\Delta}$ be the set of context-dependent fusion rules defined in Definition~\ref{def:construction-CDFG(pts)-step-simulation-rules}.
Then $\CDFG(\TM) = (Z_{\TM}, \{head\} + F_{\tg}, \{\mu\}, \{term\} + \Omega + \Lambda, P_{\TM})$
is the corresponding context-dependent fusion grammar 
with $Z_{\TM}$ and $P_{\TM}$ defined as follows.

$Z_{\TM} =  Z_{\tg} + hg(\TM)_{init} + Acc + \encsg{(\fuscomp{\tmblank})} + \sum\limits_{\substack{u \in \fsym{\Gamma}, \lambda \in \Lambda\\ 0 < i \le |Q|} \ } C(u,\lambda,i)$
where
$hg(\TM)_{init}= hg(\TM,q_{start}\sigma)$ for some arbitrary permutation $\sigma \in Q\setminus \{q_{start}\}$;
and
$Acc = ( \{v_1, \ldots, v_{|Q|+2} \}, \{ term, \fuscomp{head}, cut \},\sE,\tE,\lE)$
where
$\sE(term) = v_1 \cdots v_{|Q|}$, $\tE(term) = v_{|Q|+1}$, $\lE(term) = term$,
$\sE(\fuscomp{head}) = v_1 \cdots v_{|Q|}$, $\tE(\fuscomp{head}) = v_{|Q|+1}$, $\lE(\fuscomp{head}) = \fuscomp{head}$,
$\sE(cut) = v_{|Q|+1}$, $\tE(cut) = v_{|Q|+2}$, $\lE(cut) = cut$.
A schematic drawings of $Acc$ is depicted in Figure~\ref{fig:cdfg-tm-1-acc}.

$P_{\TM} = P_{\tg} \cup P_\Delta \cup \{  \fr(tape), accept, cut, shrink \}$
where $accept, cut$ and $shrink$ are defined as follows:
  
  $accept = (\fr(head), \{ \fr(head) \to PC_{acc} + Acc \}, \emptyset)$,
where $PC_{acc} = (Q + \{v_{tape}\},\{head,acc\},\sE,\tE,\lE)$
with $\sE(head) = q_{accept}\sigma'$, $\tE(head) = v_{tape}$, $\lE(head) = head$,
$\sE(acc) = \tE(acc) = q_{accept}$, and $\lE(acc) = acc$,
where $\sigma'$ are the states of $Q\setminus \{q_{accept}\}$ in arbitrary order.
A schematic drawings of $PC_{acc}$ is depicted in Figure~\ref{fig:cdfg-tm-1-pcacc}.

$cut = (\fr(cut), \{ \fr(cut) \to (\{v_1,v_2,v_3 \}, \{cut,\fuscomp{cut}\}, \sE,\tE, \lE)\} , \emptyset)$, where
$\sE(cut) = \sE(\fuscomp{cut}) = v_1, \tE(cut) = v_2, \tE(\fuscomp{cut}) = v_3$ and $\lE(cut) = cut,\lE(\fuscomp{cut}) = \fuscomp{cut}$, i.e., the source of the two complementary hyperedges must be the same and the targets different. The connected component is depicted in Figure~\ref{fig:cdfg-tm-1-pccut}.

$shrink = (\fr(\tmblank), \{\fr(\tmblank) \to \begin{tikzpicture}[baseline=-0.1cm,node distance=1.3cm,>=stealth',bend angle=25,auto]

  \node [inner sep = 1.5pt,circle,fill=black] at (0,0)(0) {};
  \node [inner sep = 1.5pt,circle,fill=black] at (1.5,0)(1) {};

  \node [transition] at (0.75,0)  (blank)  {$\tmblank$}
  edge [pre]                     (0)
  edge [post]                     (1);

  \node [transition] at (-0.7,0) (start)  {$\tmtapestart$}
  edge [post]                     (0);

\end{tikzpicture}
~~~~~
\begin{tikzpicture}[baseline=-0.1cm,node distance=1.3cm,>=stealth',bend angle=25,auto]

  \node [inner sep = 1.5pt,circle,fill=black] at (0,0)(0) {};
  \node [inner sep = 1.5pt,circle,fill=black] at (1.5,0)(1) {};

  \node [transition] at (0.75,0)  (blank)  {$\fuscomp{\tmblank}$}
  edge [pre]                     (0)
  edge [post]                     (1);

  \node [transition] at (2.25,0) (start)  {$\tmtapestart$}
  edge [post]                     (1);

\end{tikzpicture}\}, \emptyset)$.

\end{definition}

\begin{figure}[t]
\centering
\begin{subfigure}[b]{0.4\textwidth}
\centering
\begin{tikzpicture}[baseline=0.3cm,node distance=1.3cm,>=stealth',bend angle=25,auto]

  \node [inner sep = 1.5pt,circle,fill=black] at (0,0)(s1) {};
  \node [inner sep = 1.5pt,circle,fill=black] at (1,0)(s4) {};

  \node [inner sep = 1.5pt,circle,fill=black] at (0.5,1.75)(0) {};

  \node [inner sep = 1.5pt,circle,fill=black] at (2.5,1.75)(2) {};
  
  \node at (0.5,0) {$\dots$};    

  \node [transition] at (0,1)  (t1)  {$\fuscomp{head}$}
  edge [pre] (s1)
  edge [pre]                     (s4)
  edge [post]                    (0);

  \node[transition] at (1,1)  (t1)  {$term$}
  edge [pre]                     (s1)
  edge [pre]                     (s4)
  edge [post]                    (0);

  \node [transition] at (1.5,1.75) (cut)  {$\fuscomp{cut}$}
  edge [pre]                     (0)
  edge [post]                    (2);
               
\end{tikzpicture}
\caption{$Acc$}
\label{fig:cdfg-tm-1-acc}
\end{subfigure}~~
\begin{subfigure}[b]{0.3\textwidth}
\centering
\begin{tikzpicture}[baseline=0.3cm,node distance=1.3cm,>=stealth',bend angle=25,auto]

  \node [inner sep = 1.5pt,circle,fill=black] at (0,0)(s1) {};
  \node [inner sep = 1.5pt,circle,fill=black] at (1,0)(s4) {};

  \path (s1) edge[loop left]  node[below] {$acc$} (s1);
  
  \node [inner sep = 1.5pt,circle,fill=black] at (0.5,1.75)(0) {};
  
  \node at (0.5,0) {$\dots$};    

  \node [transition] at (0.5,1)  (t1)  {$head$}
  edge [pre] node[left] {$1$} (s1)
  edge [pre]                     (s4)
  edge [post]                    (0);
               
\end{tikzpicture}
\caption{$PC_{acc}$}
\label{fig:cdfg-tm-1-pcacc}
\end{subfigure}~~
\begin{subfigure}[b]{0.2\textwidth}
\centering
\begin{tikzpicture}[baseline=0.0cm,node distance=1.3cm,>=stealth',bend angle=45,auto]

  \node [inner sep = 1.5pt,circle,fill=black] at (0,0)(0) {};    
  \node [inner sep = 1.5pt,circle,fill=black] at (2,0)(1) {};
  \node [inner sep = 1.5pt,circle,fill=black] at (2,1)(2) {};

  \node [transition] at (1,0)  (t1)  {$cut$}
  edge [pre]                    (0)
  edge [post]                   (1);
  \node [transition] at (1,0.5)  (t3)  {$\fuscomp{cut}$}
  edge [pre]                    (0)
  edge [post]                   (2);
  
\end{tikzpicture}
\caption{$PC_{cut}$}
\label{fig:cdfg-tm-1-pccut}
\end{subfigure}
\caption{Schematic drawings of connected components used for acceptance and for disconnecting the terminal and marker labeled string graph.}
\label{Fig:cdfg-TM-components-acc}
\end{figure}
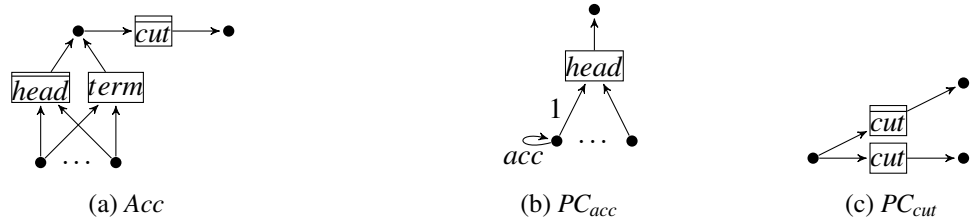

\begin{example}
Reconsidering the hypergraph $H' = hg(q_{accept} q_{aux} q_{start}, dc, b, ab)$ depicted in the upper right in Figure~\ref{fig:cdfg-tm-example-step-simulation}.
Then it is easy to see that
$H' + Acc \dder_{accept} H'' \dder_{cut} H'''$
where $H''$ and $H'''$ are depicted in Figure~\ref{Fig:cdfg-tm-example-accepting}.
Note that, the upper part of $H'''$ is terminal and marker labeled only.
\begin{figure}
\begin{subfigure}[b]{0.45\textwidth}
\scalebox{0.8}{\begin{tikzpicture}[->,>=stealth',shorten >=1pt,auto,node distance=2.8cm,
                    semithick]
  \tikzstyle{every state}=[draw=black,text=black]

    \foreach \x in {-1,1,3,5}{
    \node [inner sep = 1.5pt,circle,fill=black] at (\x,0)(\x) {};
  }
  \foreach \x in {2,4,6}{
     \node [inner sep = 1.5pt,circle,fill=black] at (\x,1)(\x) {};
  }
 
  \node [transition] at (2.5,0.5)  (t2)  {$cut$}
  edge [pre]                    (3)
  edge [post]                   (2);
  \path (2) edge [loop left] node{$\mu$}();

  \node [inner sep = 1.5pt,circle,fill=black] at (4.5,0.5)(cutout) {};

  \node [transition] at (3.5,0.5)  (t2)  {$\fuscomp{cut}$}
  edge [pre]                    (3)
  edge [post]                   (cutout);

  \foreach \x/\z [evaluate=\x as \y using \x+2] in {2/a,4/b}{
    \path (\x) edge[->,shorten >= 5pt] node[above=.06]{$\z$} (\y);
  }
  \foreach \x/\z [evaluate=\x as \y using \x+2,evaluate=\x as \p using \x+1] in {-1/\fsym{d},1/\fsym{c},3/\fsym{b}}{    
    \node [transition] at (\p,0)  (t\x)  {$\z$}
      edge [pre]                    (\x)
      edge [post,shorten >= 5pt]                   (\y);
  }

  \node[state,accepting]   (acc)  [below = 1.4 of 1]      {$q_{accept}$};

  \path (acc) edge[loop left]  node[left] {$acc$} (acc);

  \node[state] [right of=acc] (start)                    {$q_{start}$};
  \node[state]         (q1) [right of=start] {$q_{aux}$};

  \node [transition,dashed] [below = 0.3 of 1]  (head)  {$term$}
  edge [pre,dashed] node [right = 0.01, very near start] {$3$}                    (start)
  edge [pre,bend left=28,dashed] node [below = 0.01, very near start] {$2$}                  (q1)
  edge [pre,dashed] node [below = 0.01, very near start] {$1~~~$}                  (acc)
  edge [post,dashed]                    (3);

  \path (start) edge[bend left=8]              node {$b/\tmblank/r$} (q1);
  \path (q1) edge[bend left=8]              node {$b/b/n$} (start);
  \path (start) edge              node {$a/c/r$} (acc);

\end{tikzpicture}}
\caption{$H''$: hypergraph representation of a terminated Turing machine after applying $accept$}
\label{}
\end{subfigure}~~~~
\begin{subfigure}[b]{0.45\textwidth}
\scalebox{0.8}{\begin{tikzpicture}[->,>=stealth',shorten >=1pt,auto,node distance=2.8cm,
                    semithick]
  \tikzstyle{every state}=[draw=black,text=black]

    \foreach \x in {-1,1,3,5}{
    \node [inner sep = 1.5pt,circle,fill=black] at (\x,0)(\x) {};
  }
  \foreach \x in {2,4,6}{
     \node [inner sep = 1.5pt,circle,fill=black] at (\x,1)(\x) {};
  }
 
  \path (2) edge [loop left] node{ $\mu$}();
  
  \foreach \x/\z [evaluate=\x as \y using \x+2] in {2/a,4/b}{
    \path (\x) edge[->,shorten >= 5pt] node[above=.06]{$\z$} (\y);
  }
  \foreach \x/\z [evaluate=\x as \y using \x+2,evaluate=\x as \p using \x+1] in {-1/\fsym{d},1/\fsym{c},3/\fsym{b}}{    
    \node [transition] at (\p,0)  (t\x)  {$\z$}
      edge [pre]                    (\x)
      edge [post,shorten >= 5pt]                   (\y);
  }

  \node[state,accepting]   (acc)  [below = 1.4 of 1]      {$q_{accept}$};
  \path (acc) edge[loop left]  node[left] {$acc$} (acc);

  \node[state] [right of=acc] (start)                    {$q_{start}$};
  \node[state]         (q1) [right of=start] {$q_{aux}$};

  \node [transition,dashed] [below = 0.3 of 1]  (head)  {$term$}
  edge [pre,dashed] node [right = 0.01, very near start] {$3$}                    (start)
  edge [pre,bend left=28,dashed] node [below = 0.01, very near start] {$2$}                  (q1)
  edge [pre,dashed] node [below = 0.01, very near start] {$1~~~$}                  (acc)
  edge [post,dashed]                    (3);

  \path (start) edge[bend left=8]              node {$b/\tmblank/r$} (q1);
  \path (q1) edge[bend left=8]              node {$b/b/n$} (start);
  \path (start) edge              node {$a/c/r$} (acc);

\end{tikzpicture}}
\caption{$H'''$: two connected components obtained by applying $cut$ to $H''$ }
\label{}
\end{subfigure}
\caption{Connected components obtained after the application of $accept$ and $cut$.}
\label{Fig:cdfg-tm-example-accepting}
\end{figure}
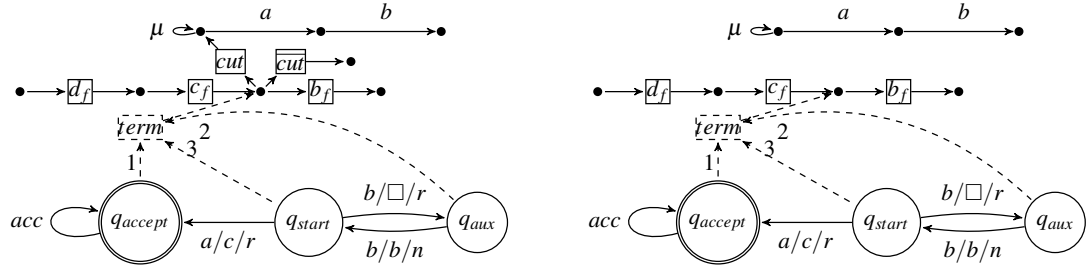
\end{example}


Our main theorem is that the language recognized by some Turing machine and the generated language of the corresponding context-dependent fusion grammar coincide up to representation of strings as graphs.

\begin{theorem}\label{thm:cdfg-tm}
$L(\CDFG(\TM)) = \{ \sg(w) \mid w \in L(\TM)\}$.
\end{theorem}

The proof is based on the following lemmata.

\begin{lemma}\label{lemma:cdfg-tm-ci-step}
Let $c_{i} = (p, \alpha, \beta), c_{i+1} = (q, \alpha', \beta') \in \mathit{conf}(\TM), \delta = (p,x,y,dir,q) \in \Delta$, $w \in \Omega^*$, and let
$pr(s(head), k) = q$ for the $head$-labeled hyperedge in $H = hg(p q_2 \cdots q \cdots q_{|Q|}, \alpha, \beta, w)$.
Then
$ c_{i} \tmstep{\TM} c_{i+1} \text{ wrt } \delta \text{ if and only if } H + C(u,x/y/dir, k) + tape_{\tmtapestart} + tape_{\tmtapeend} + \encsg{(\fuscomp{\tmblank})}\dder^* H'$,
where $u$ is the last symbol of $\alpha$ if $\alpha \ne \varepsilon$ and $\tmblank$ otherwise, and $H' =  hg(q q_2 \cdots p \cdots q_{|Q|}, \alpha', \beta', w)$.
\end{lemma}

\begin{proof}
If $c_i = c_{i+1}$, then the one-to-one correspondence trivially holds.
Assume that $c_{i} \tmstep{\TM} c_{i+1} \ne c_i$ wrt $\delta$ holds. We distinguish the following six cases.

Case 1) $(p, \alpha u, x \beta) \tmstep{\TM} (q, \alpha, u y \beta) \text{ wrt } \delta = (p,x,y,l,q)$.
Then we can restrict the hypergraph to $H + C(u,x/y/l, k)$ (by multiplying unneeded connected components by 0) and apply the derivation
$H + C(u,x/y/l, k) \dder_{\Delta(u,x/y/l)} X \dder_{fuse\_loop\_in(u)} X' \dder_{fuse\_2out(x)} H'$
due to the following reasoning.
Because $u$ is the symbol left of $x$ and $pr(s(head), k) = q$
the rule $\Delta(u,x/y/dir)$ can be applied matching the $head$-hyperedge in $H$ and the $\fuscomp{head}$-hyperedge in $C(u,x/y/dir, k)$ yielding the connected component $X$ which is $H$ extended by the respective hypergraphs $hg(u,x/y/l)$, where $v \in V_{hg(u,x/y/l)}$ and $v_{tape} \in V_{H}$ are identified.
Because of the additional edges $fuse\_loop\_in(u)$ becomes applicable and the derivation $X \dder_{fuse\_loop\_in(u)} X'$ fuses the previously added $\fuscomp{u}$-loop and the in-going $u$-edge in $X$.
Afterwards, $fuse\_2out(x)$ is applicable and the derivation $X' \dder_{fuse\_2out(x)} H'$ yields the requested hypergraph due to the fact that the two out-going $\fsym{x}$- and $\fsym{\fuscomp{x}}$-hyperedges are fused. 

The other cases use similar arguments and are therefore stated less explicit.

Case 2) $(p, \varepsilon, x \beta)  \tmstep{\TM} (q, \varepsilon, \tmblank y \beta) \text{ wrt } \delta = (p,x,y,l,q)$.
Then we can restrict the hypergraph to the three connected components $H + C(\tmblank,x/y/l, k) + tape_{\tmtapestart}$ and apply the derivation
$H + C(\tmblank,x/y/l, k) +  tape_{\tmtapestart} \dder_{\fr(\tmtapestart)} \tilde{H} + C(\tmblank,x/y/l, k) \dder_{\Delta(\tmblank,x/y/l)} X \dder_{fuse\_loop\_in(\tmblank)} X' \dder_{fuse\_2out(x)} H'$,
where $\tilde{H}$ is $H$ but $\alpha = \tmblank$ instead of $\varepsilon$;
afterwards the same reasoning as the previous case is applied.

Case 3) $(p, \alpha, \varepsilon)  \tmstep{\TM} (q, \alpha, y) \text{ wrt } \delta = (p,\tmblank,y,l,q)$. In the subcase $\alpha \ne \varepsilon$ we have the derivation
$H + C(\tmblank,x/y/l, k) + tape_{\tmtapeend} \dder_{\fr(\tmtapeend)} \tilde{H} + C(\tmblank,\tmblank/y/l, k) \dder_{\Delta(\tmblank,x/y/l)} X \dder_{fuse\_loop\_in(\tmblank)} X' \dder_{fuse\_2out(x)} H'$
similar to the previous case
and in the subcase $\alpha = \varepsilon$ all five connected components are needed and the derivation of the previous subcase is prepended by $\dder_{\fr(\tmtapestart)}$ and appended by $\dder_{shrink}$.

Case 4) $(p, \alpha, x \beta)  \tmstep{\TM} (q, \alpha, y \beta) \text{ wrt } \delta = (p,x,y,n,q)$. Then we have the derivation
$H + C(u,x/y/n, k)\allowbreak \dder_{\Delta(u,x/y/n)} X \dder_{fuse\_2out(x)} H'$.

Case 5) $(p, \alpha, x \beta)  \tmstep{\TM} (q, \alpha y, \beta) \text{ wrt } \delta = (p,x,y,r,q)$.
Subcase $\alpha \ne \varepsilon$ is analog to Case 1 yielding the derivation
$H + C(u,x/y/r, k) \dder_{\Delta(u,x/y/r)} X \dder_{fuse\_loop\_out(u)} X' \dder_{fuse\_2in(x)} H'$
and subcase $\alpha = \varepsilon$ is analog to the respective subcase in Case 3.

Case 6) $(p, \alpha, \varepsilon)  \tmstep{\TM} (q, \alpha y, \varepsilon) \text{ wrt } \delta = (p,\tmblank,y,r,q)$ is analog to Case 3 wrt both cases.

Conversely, given $H + C(u,x/y/dir, k) + tape_{\tmtapestart} + tape_{\tmtapeend} + \encsg{(\fuscomp{\tmblank})} \dder^* H' \ne H$.
This derivation can be reduced to the six cases above.
The applicability of $\Delta(u,x/y/dir)$ to $H + C(u,x/y/dir, k)$
implies that there exists an $x/y/dir$-edge between some vertices $p,q \in V_{H}$ and that the current symbol read is $x$, where $pr(s(head), 1) = p$ and $pr(s(head), k) = q$ for $head \in E_{H}, \lE(head) = head$.
Furthermore, the construction of $\Delta(u,x/y/dir)$ and $C(u,x/y/dir, k)$ gives that $pr(s(head), 0) = q$ for the reconstructed hyperedge  $head \in E_{H'}$.
This implies $c_{i} \tmstep{\TM} c_{i+1}$ wrt $\delta = (p,x,y,dir,q)$.
\end{proof}

\begin{lemma}\label{lemma:cdfg-tm-computation}
$c_0 \tmstep{\TM}^k c_{k}  = (q_k, \alpha, \beta)$ wrt input $w$ implies $Z \dder^* Z + hg(q_k \sigma, \alpha,\beta, w)$ for some $\sigma \in Q \setminus \{q_k\}$.
\end{lemma}

\begin{proof}
Induction base: $k=0$.
For each $w=w_1\cdots w_n \in \Omega^*$ exists a derivation in $\CDFG(\TM)$ such that
$hg(\TM)_{init} + w_{start} + w_{end} + \sum\limits_{i=1}^n w_{w_i} \dder^{n+2} H_0 = hg(q_{start}\sigma, \varepsilon, w, w)$
using $\fr(tape)$ once and $\fr(gen)$ $n +1$-times due to Proposition~\ref{prop:cdfg-tape-graph-1}.
Consequently,
$Z \dder_m Z + hg(\TM)_{init} + w_{start} + w_{end} + \sum\limits_{i=1}^n w_{w_i}$ $\dder^{n+2} Z + H_0.$

Induction step: $c_0 \tmstep{\TM}^{k+1} c_{k+1}$ implies $c_0 \tmstep{\TM}^{k} c_{k} \tmstep{\TM} c_{k+1}$ for some $c_k = (q_k, \alpha',\beta')$.
By induction hypothesis $c_0 \tmstep{\TM}^{k} c_{k}$ implies $Z \dder^* Z + H_k$,
where $H_k = hg(q_k \sigma', \alpha',\beta', w)$ for some $\sigma'$.
Let $c_{k} \tmstep{\TM} c_{k+1}$ be wrt $\delta = (q_k,x,y,dir,q_{k+1})$.
By construction of $\CDFG(\TM)$ there exists $C(u,x/y/dir, j) \in \calC(Z)$
for each $(p,x,y,dir,q) \in \Delta, u \in \fsym{\Gamma}, 1\le j\le |Q|$.
Consequently, such a connected component also exists for $\delta$ such that $pr(s(head), 1) = q_{k}$ and $pr(s(head), j) = q_{k+1}$.
Let $C_\delta$ be this suitable connected component.
Then there is a derivation $Z + H_k \dder_m Z + H_k + C_\delta + tape_{\tmtapestart} + tape_{\tmtapeend} + \encsg{(\fuscomp{\tmblank})} \dder^* Z + H_{k+1}$, where $m$ is a multiplication, $H' =  hg(q_{k+1}\sigma, \alpha, \beta, w)$ and $\sigma$ is the same as $\sigma'$ except that $q_{k+1}$ is replaced by $q_k$, due to Lemma~\ref{lemma:cdfg-tm-ci-step}.
\end{proof}

\begin{definition}
  $hg(\TM,\sigma, \alpha, \beta,w )_{acc}$ is a connected component isomorphic to
  $hg(\TM,\sigma, \alpha, \beta,w )$ but the label of the $head$-hyperedge is $term$.
\end{definition}

\begin{proof} of Theorem~\ref{thm:cdfg-tm}.
  We show first $w \in L(\TM)$ implies $\sg(w) \in L(\CDFG(\TM))$.
Let $w \in L(\TM)$. Then
$c_0 = (q_{start}, \varepsilon, w) \tmstep{\TM}^* c_a = (q_{accept}, \alpha,\beta)$ for some $w \in \Omega^*, \alpha,\beta \in \Gamma^*$.
Then
\begin{align*}
  Z
  & \dder^* ~~Z~ + hg(q_{accept} \sigma, \alpha, \beta,w) ~~~\text{ due to Lemma~\ref{lemma:cdfg-tm-computation}} \\
  & \dder_m Acc + hg(q_{accept} \sigma, \alpha, \beta,w) ~~~\text{ where } m(x) = \begin{cases} 1 & x\in\{Acc,hg(q_{accept} \sigma, \alpha, \beta,w)\}\\ 0 & \text{otherwise} \end{cases}\\
  & \dder_{accept} hg(q_{accept}\sigma,\alpha, \beta,w)_{acc}\\
  & \dder_{cut} (hg(\TM,q_{accept} \sigma) + \encsg{(\alpha, \beta)})/_{begin(\sg(\beta)) \equiv \tE(head)} + \sg(w)_\mu.
 \end{align*}
Hence, $\sg(w) \in L(\CDFG(\TM))$.

\bigskip

The converse is more complicated to show. 
$\sg(w) \in L(\CDFG(\TM))$ means there is a derivation $Z \dder^* H$ with $Y \in \calC(H), H \in \Htm{\{\mu, term \}}{\Lambda + \Omega} -  \Ht{ \{term\} + \Lambda + \Omega}$ and $rem_\mu(Y) = \sg(w)$.
Without loss of generality, one can assume:
\begin{itemize}
\item $H = Y$ because the other connected components can be multiplied by 0.
\item There is exactly one marker component in each derived hypergraph because two marked derived hypergraphs can never be fused with each other.
\item The set of sources of the hypergraph representation of a Turing machine (and extended connected components) is $Q$ because source vertices of these connected components cannot be fused with each other.
\item All necessary multiplications are done as first derivation step and all applications of context-free fusion rules ($\fr(gen)$, $\fr(\tmtapestart), \fr(\tmtapeend)$ and $\fr(tape)$), are done before any application of some context-dependent fusion rule with context conditions.
\end{itemize}
Moreover, some of the rules are sequentially dependent with respect to the same connected component.
\begin{enumerate}
\item $cut$ and $accept$ are sequentially dependent, because the $cut$-hyperedge required in positive context condition of the $cut$-rule is added to $v_{tape}$ by the application of $accept$ to $hg(q_{accept} \sigma, \alpha, \beta,w) + Acc$.
\item $\Delta(u,\lambda)$ and $accept$ are sequentially dependent, because $q_{start} \ne q_{accept}$.

\item $\Delta(u,\lambda)$ and and some rule $r$ in the latter set in $P_{\Delta}$ are sequentially dependent, because the complementary edges attached to $v_{tape}$ required by $r$ are attached by the application of $\Delta(u,\lambda)$.

\item $\fr(gen), \fr(tape)$ and $\Delta(u,\lambda)$ are sequentially dependent, because only if the tape graph is attached to the $head$-hyperedge of some hypergraph derived from $hg(\TM)_{init}$, then the positive context conditions of $\Delta(u,\lambda)$ are satisfied.
\end{enumerate}
Furthermore, the positive and negative context conditions restrict the fusion process dramatically.
\begin{enumerate}

\item $cut$ is only applicable to $hg(q_{accept}\sigma,\alpha, \beta,w)_{acc}$ for arbitrary $\sigma,\alpha, \beta,w$.

\item No two $\Delta(u_1,\lambda_1)$ and $\Delta(u_2,\lambda_2)$ are applicable to some hypergraph representation of some configuration directly one after the other, because the application of $\Delta(u_1,x_1/y_1/dir_1)$ attaches
a $\fuscomp{u_1}$- (if $dir_1 \in \{l,r\}$) and a $\fuscomp{x_1}$-hyperedge to $v_{tape}$, hence, the negative context conditions of $\Delta(u_2,\lambda_2)$ are not satisfied. 

\item $\Delta(u,\lambda)$ may only be applicable to some $hg(q_1 \cdots q_{|Q|}, \alpha u, x\beta,w)$ and $C(u,\lambda,j)$, where $\lambda = x/y/dir$ and $q_1 \cdots q_{|Q|},\alpha, u, x,\beta',w$ are arbitrary.
No fusion is possible wrt $C(u_1,\lambda_1,j) + C(u_2,\lambda_2,k) + Acc$.

\item A rule of the latter set in $P_{\Delta}$ is only applicable to the connected component obtained by the fusion wrt $\Delta(u,\lambda)$. No fusion is possible inside $C(u,\lambda,k)$.

\end{enumerate}
The last argument is that the context-dependent fusion rule $accept$ can only be applied if there exists a match into some connected component derived from $hg(\TM)_{init}$ and $Acc$. The restriction to $hg(\TM)_{init}$ comes from the fact, that the $head$-hyperedge must not be part of some $C$ connected component.
Hence, $accept$ and $cut$ are only applicable to a hypergraph representation of a configuration wrt input $w$ if and only if $w \in L(\TM)$. Moreover, $accept$ and $cut$ can be delayed to the end of the derivation.

Let $C_1, \ldots, C_n$ be the $C$-components in the order in which they are used in the derivation $Z \dder^* \sg(w)_\mu$. Then, using the remarks above, one can rearrange the derivation such that is is of the form
\begin{align*}
  Z
  & \dder_m X_0 \dder^*_{\fr(gen)} X_1 \dder^*_{\fr(\tmtapestart)} X_2 \dder^*_{\fr(\tmtapeend)} X_3 \dder^*_{\fr(tape)} hg(\sigma_0, \alpha_0, \beta_0, w) + Acc + \sum\limits_{i=1}^n C_i \dder^* hg(\sigma_1, \alpha_1, \beta_1, w) + Acc + \sum\limits_{i=2}^n C_i\\
  &\dder^* \ldots \dder^* hg(\sigma_n, \alpha_n, \beta_n, w) + Acc \dder_{accept} Y_n \dder_{cut} Y'_n + \sg(w)_\mu \dder_{m_0} \sg(w)_\mu.
\end{align*} 
Consequently, due to Lemma~\ref{lemma:cdfg-tm-ci-step}\footnote{The cases where $\fr(\tmtapestart), \fr(\tmtapeend), shrink$ are applied do not occur due to the assumption that all context-free fusion rules are applied first.}, this implies $c_0 \tmstep{\TM}^* (q_{accept}, \alpha_n, \beta_n)$ wrt input $w$.
Hence,
$w \in L(\TM)$.
\end{proof}

\section{Conclusion}
\label{sec:conclusion}
In this paper, we have continued the research on context-dependent fusion grammars by transforming Turing machines into this type of hypergraph grammars.
This reduction gives us interesting insights into these grammars because the transformation proves that context-dependent fusion grammars are another universal computing model and can generate all recursive enumerable string languages (up to representation).
Note that a similar construction also works for computation of partial functions. In this case the connected components $tape_{start}, tape_{end}$ and $tape_x$ are replaced by a tape graph representing the Turing machines input $x_1\ldots x_n \in \Sigma^*$, where the start is attached to some $tape$-hyperedge.
However, further research is needed including the following open question.
In the literature, one encounters model transformations from several modeling approaches into Turing machines.
Now they can be extended to context-dependent fusion grammars.
Does this provide interesting insights?
Are only positive or only negative context conditions powerful enough to cover Turing machines?
How does a natural transformation of context-dependent fusion grammars into splicing/fusion grammars or the other way round look like?


\bibliographystyle{eptcs}
\bibliography{lit,litall}
\end{document}